\documentclass[11pt]{article}
\usepackage{graphicx} % Required for inserting images
\usepackage{fullpage}
\usepackage{multirow} % Needed for multirow cells
\usepackage{array}    % Optional, for improved array handling
\usepackage{enumitem}
\usepackage[backend=biber,style=authoryear-comp,natbib=true,maxcitenames=9,maxbibnames=9]{biblatex}
\usepackage{hyperref}
\usepackage{booktabs}  % for nicer table rules (optional)
\usepackage{diagbox} % For diagonal split cells
\usepackage{nicefrac}
\usepackage{amssymb}
\usepackage{amsmath}
\usepackage{amsthm}
\usepackage{bbm}
\allowdisplaybreaks % allow displayed math to break
\newtheorem{theorem}{Theorem}[section]

\newtheorem{corollary}[theorem]{Corollary}
\newtheorem{lemma}[theorem]{Lemma}
\newtheorem{proposition}[theorem]{Proposition}
\newtheorem{claim}[theorem]{Claim}
\newtheorem{observation}[theorem]{Observation}

\theoremstyle{definition}
\newtheorem{definition}[theorem]{Definition}
\newtheorem{example}[theorem]{Example}
\newtheorem{remark}[theorem]{Remark}
\usepackage{url}
\usepackage{MnSymbol} % for bigudot
\usepackage[table]{xcolor}

\usepackage{tikz}
\usetikzlibrary{plotmarks}
\usetikzlibrary{decorations.pathreplacing,calligraphy}

\newcommand{\indicator}[1]{\mathbbm{1}\left[{#1}\right]}
\newcommand{\reals}{\mathbb{R}}
\newcommand{\contract}{\vec{\alpha}}
\newcommand{\nnreals}{\mathbb{R}_{\geq 0}}

\newcommand{\Ex}[2][]{\mbox{\rm\bf E}_{#1}\left[#2\right]}
\renewcommand{\Pr}[2][]{\mbox{\rm\bf Pr}_{#1}\left[#2\right]}

\newcommand{\Seq}{S^\dagger}
\newcommand{\SSeq}[1]{\mathcal{D}^\dagger}

\usepackage{algorithm}
\usepackage{algpseudocode}

\newcommand{\E}{\mathbf{E}}

\title{Black-Box Lifting and Robustness Theorems for\\Multi-Agent Contracts%A Black-Box Lifting Theorem for Multi-Agent Contracts%
\thanks{This project has been partially funded by the European Research Council (ERC) under the European Union's Horizon Europe Program (FACT, grant agreement No. 101170373), by an Amazon Research Award, by the NSF-BSF (grant number 2020788), by the Israel Science Foundation Breakthrough Program (grant No.~2600/24), and by a grant from TAU Center for AI and Data Science (TAD). T. Ezra is supported by the Harvard University Center of Mathematical Sciences and Applications.}
}
\author{Paul D\"utting\thanks{Google Research, Z\"urich, Switzerland. Email: \texttt{duetting@google.com}} \and Tomer Ezra\thanks{Harvard University, Cambridge, USA. Email: \texttt{tomer@cmsa.fas.harvard.edu}} \and Michal Feldman\thanks{Tel Aviv University, Tel Aviv, Israel. Email: \texttt{mfeldman@tauex.tau.ac.il}} \and Thomas Kesselheim\thanks{University of Bonn, Bonn, Germany. Email: \texttt{thomas.kesselheim@uni-bonn.de}}}
\date{}

\begin{document}

\maketitle

\begin{abstract}
Multi-agent contract design has largely evaluated contracts through the lens of pure Nash equilibria (PNE). This focus, however, is \emph{not} without loss: In general, the principal can strictly gain by recommending a complex, possibly correlated, distribution over actions, while preserving incentive compatibility.
In this work, we extend the analysis of multi-agent contracts beyond pure Nash equilibria to encompass more general equilibrium notions, including mixed Nash equilibria as well as (coarse-)correlated equilibria (CCE). The latter, in particular, captures the limiting outcome of agents engaged in learning dynamics.
    
Our main result shows that for submodular and, more generally, XOS rewards, such complex recommendations yield at most a constant-factor gain: there exists a contract and a
PNE whose utility is within a constant factor of the best CCE achievable by any contract. This provides a black-box \emph{lifting}: results established against the best PNE automatically apply with respect to the best CCE, with only a constant factor loss.
For submodular rewards, we further show how to transform a contract and a PNE of that contract 
into a new contract such that any of its CCEs gives a constant approximation to the PNE. This yields black-box \emph{robustness}: up to constant factors, guarantees established for a specific contract and PNE automatically extend to the modified contract and any of its CCEs.
We thus expand prior guarantees for multi-agent contracts and lower the barrier to new ones. 
As an important corollary, we obtain poly-time algorithms for submodular rewards that achieve constant approximations in any CCE, against the best CCE under the best contract. 
Such worst-case guarantees are provably unattainable for XOS rewards. 
Finally, we bound the gap between different equilibrium notions for subadditive, supermodular, and general rewards.
\end{abstract}

\section{Introduction}

A classic contract setting features a \emph{principal} who must incentivize a self-interested \emph{agent} to take costly actions that generate value for the principal. Examples include a firm designing incentive schemes for employees or a platform rewarding contributors. The central difficulty is that actions are typically \emph{hidden}: the principal cannot directly monitor actions and must instead align incentives through payments that depend only on observable outcomes.

In recent years, many such relationships have migrated to computational platforms. These environments are large-scale and complex, calling for an algorithmic treatment. \emph{Algorithmic contract design} has therefore emerged at the intersection of economics and computation, developing models and schemes for designing incentives in these settings (we refer to \cite{DuttingFT24survey} for a comprehensive survey).

A natural focal point in this literature is \emph{combinatorial contracts}: the principal may interact with teams of agents, a single agent may choose among combinations of actions, and the overall reward depends on the joint selection through a combinatorial reward function that can exhibit both substitutability and complementarity \citep[e.g.,][]{BFN06,DEFK23,CastiglioniMN23,DEFK21,VDPP24,EFS24,DEFK25,HannCS24,DuttingFT24,DuttingFTR26,Feldman25}. The principal's objective is to select a contract that maximizes her expected utility given the agents' equilibrium behavior.

A key limitation of much of this literature is its (near) exclusive focus on \emph{pure Nash equilibria} (PNE). In multi-agent settings this restriction is \emph{not} without loss of generality. As observed by \citet{BabaioffFN10}, the gap can arise already for submodular rewards (see Example~\ref{ex:mixed}). In particular, a principal can sometimes do strictly better by inducing agents to play a \emph{distribution} over actions from which no agent wishes to deviate. This raises a natural question: \emph{how much additional utility can the principal obtain by moving beyond PNE to richer equilibrium concepts that allow randomization or correlation?}
This question is especially salient because natural learning dynamics typically converge to these broader, and often correlated, equilibrium notions.

\subsection{Model and Research Problem}

We consider the multi-agent combinatorial contracts model of \cite{DEFK25}, which generalizes the models of \cite{DEFK21} and \cite{DEFK23}. 
In the multi-agent combinatorial contracts model,
a single principal interacts with $n$ agents.  
Each agent $i \in [n]$ can take any subset of actions $S_i$ from an available action set $A_i$.
We let $A = \bigcupdot_i A_i$, and use $m = |A|$ to denote the total number of actions. We refer to the special case where $|A_i| = 1$ for all $i \in [n]$, i.e., each agent can either take action or not, as the binary-actions case. 
The agents' choice of actions $S \subseteq A$, with $S_i = S \cap A_i$ the set of actions chosen by agent $i \in [n]$, determines a reward $f(S)$ through a reward function $f: 2^A \rightarrow \reals_{\geq 0}$. We generally assume that the reward function is normalized (so that $f(\emptyset) = 0$) and monotone (so that $S \subseteq T$ implies $f(S) \leq f(T)$). 
Access to $f$ may be given through a \emph{value oracle} or a \emph{demand oracle} (see Section~\ref{sec:model}).  
The agents each have a cost $c_j$ for each action $j \in A$, and the cost for a set of actions $S_i \subseteq A_i$ is $c(S_i) = \sum_{j \in S_i} c_j$.
To incentivize the agents, the principal designs a linear contract $\contract = (\alpha_1, \ldots, \alpha_n) \in [0,1]^n$.  The interpretation is that the principal 
pays each agent $i$ an $\alpha_i$ fraction of the reward $f(S)$ that results from the agents' actions $S$. The principal's utility is $(1-\sum_{i} \alpha_i)f(S)$. The agents, in turn, have a utility of $\alpha_i f(S) - c(S_i)$. Since the agents' utilities depend on each other, we are interested in equilibria among agents. 

While prior work has focused on pure Nash equilibria (PNE), here we are interested in exploring the more general equilibrium concepts of mixed Nash equilibria (MNE), correlated equilibria (CE), and coarse-correlated equilibria (CCE). 
In these more general equilibrium notions the principal induces the agents to play a distribution over actions, while ensuring that no agent has an incentive to deviate.
The randomization can be either independent (as in MNE), or correlated (as in CE and CCE), where CE and CCE differ in how they capture the property that no agent wants to deviate (for formal definitions see Section~\ref{sec:model}).
It is well known that these equilibrium concepts are successive generalizations, i.e., 
\[
\text{PNE $\subseteq$ MNE $\subseteq$ CE $\subseteq$ CCE}. 
\]

Our goal in this work is to bound the gap in the principal's utility between different equilibria under different equilibrium notions, potentially achieved by different contracts, and provide algorithms for computing contracts with strong robustness guarantees. We consider both rewards from the hierarchy of complement-free set functions (submodular $\subseteq$ XOS $\subseteq$ subadditive) as well as supermodular and general monotone rewards.

\subsection{Our Contribution}

We discuss our main results, for submodular and XOS reward functions, 
in Section~\ref{sec:results-main}.
In Section~\ref{sec:results-additional}, we cover additional results that map the broader landscape of gaps between different equilibrium concepts under prominent classes of reward functions. We provide an illustration and an overview of our results in Figure~\ref{fig:gaps} and Table~\ref{tab:equilibrium_results}.

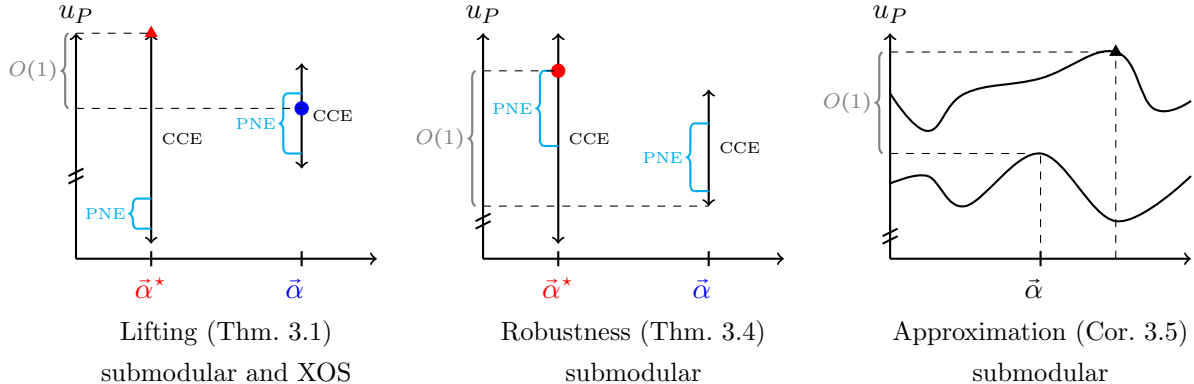
\begin{figure}
\begin{center}
%% LEFT FIGURE
\begin{minipage}{0.32\textwidth}
\begin{tikzpicture}
\draw[->,thick] (0,0) -- (0,3) node[above] {$u_P$};
\draw[->,thick] (0,0) -- (4,0);
\draw[-,thick] (1,0.1) -- (1,-0.1) node[below] {\textcolor{red}{$\contract^\star$}};
\draw[<->,thick] (1,0.2) --node[right]{\tiny CCE} (1,3);
\node[mark size=2.5pt,color=red] at (1,3) {\pgfuseplotmark{triangle*}};
\draw [decorate,
    decoration = {brace}, thick, cyan] (0.8,0.4) --node[left]{\tiny \textcolor{cyan}{PNE}} (0.8,0.8);
\draw[-,thick,cyan] (1,0.4) -- (0.8,0.4);
\draw[-,thick,cyan] (0.8,0.8) -- (1,0.8);
\draw[<->,thick] (3,1.2) --node[right]{\tiny CCE} (3,2.6);
\draw [decorate,
    decoration = {brace}, thick, cyan] (2.8,1.4) --node[left]{\tiny \textcolor{cyan}{PNE}} (2.8,2.2);
\draw[-,thick,cyan] (3,1.4) -- (2.8,1.4);
\draw[-,thick,cyan] (2.8,2.2) -- (3,2.2);
\filldraw[blue] (3,2.0) circle (2.5pt);
\draw[-,thick] (3,0.1) -- (3,-0.1) node[below] {$\textcolor{blue}{\contract}^{\textcolor{white}{\star}}$};
\draw[-,dashed] (0,3)--(1,3);
\draw[-,dashed] (0,2)--(3,2);
\draw [decorate,
    decoration = {brace}, thick, gray] (-0.1,2.0) --node[left]{\scriptsize $O(1)$} (-0.1,3.0);
\node at (2,-1) {\small Lifting (Thm.~\ref{thm:constant-gap})};
\node at (2,-1.5) {\small submodular and XOS};
\draw[-,thick] (-0.1,1) -- (0.1,1.1);
\draw[-,thick] (-0.1,1.1) -- (0.1,1.2);
\end{tikzpicture}
\end{minipage}
%% MIDDLE FIGURE
\begin{minipage}{0.32\textwidth}
\begin{tikzpicture}
\draw[->,thick] (0,0) -- (0,3)  node[above] {$u_P$};;
\draw[->,thick] (0,0) -- (4,0);
\draw[-,thick] (1,0.1) -- (1,-0.1) node[below] {\textcolor{red}{$\contract^\star$}};
\draw[<->,thick] (1,0.2) --node[right]{\tiny CCE} (1,3);
\draw [decorate,
    decoration = {brace}, thick, cyan] (0.8,1.5) --node[left]{\tiny \textcolor{cyan}{PNE}} (0.8,2.5);
\draw[-,thick,cyan] (1,1.5) -- (0.8,1.5);
\draw[-,thick,cyan] (0.8,2.5) -- (1,2.5);
\filldraw[red] (1,2.5) circle (2.5pt);
\draw[-,thick] (3,0.1) -- (3,-0.1) node[below] {$\textcolor{blue}{\contract}^{\textcolor{white}{\star}}$};
\draw[<->,thick] (3,0.7) --node[right]{\tiny CCE} (3,2.25);
\draw [decorate,
    decoration = {brace}, thick, cyan] (2.8,0.9) --node[left]{\tiny \textcolor{cyan}{PNE}} (2.8,1.8);
\draw[-,thick,cyan] (3,0.9) -- (2.8,0.9);
\draw[-,thick,cyan] (2.8,1.8) -- (3,1.8);
\draw[-,dashed] (0,2.5) -- (1,2.5);
\draw[-,dashed] (0,0.7) -- (3,0.7);
\draw [decorate,
    decoration = {brace}, thick, gray] (-0.1,0.7) --node[left]{\scriptsize $O(1)$} (-0.1,2.5);
\node at (2,-1) {\small Robustness (Thm.~\ref{thm:constant-gap-sub-CCE})};
\node at (2,-1.5) {\small submodular};
\draw[-,thick] (-0.1,0.5) -- (0.1,0.6);
\draw[-,thick] (-0.1,0.4) -- (0.1,0.5);
\end{tikzpicture}
\end{minipage}
%% RIGHT FIGURE
\begin{minipage}{0.32\textwidth}
\begin{tikzpicture}
\draw[->,thick] (0,0) -- (0,3)  node[above] {$u_P$};;
\draw[->,thick] (0,0) -- (4,0);
\draw[-,thick] (2,0.1) -- (2,-0.1) node[below] {$\contract^{\textcolor{white}{\star}}$};
\draw [-,thick] plot [smooth, tension=0.7] coordinates { (0,2.2) (0.5,1.7) (1,2.2) (2,2.4) (3,2.75) (3.5,2) (4,2.1)};
\node[mark size=2.5pt] at (3,2.75) {\pgfuseplotmark{triangle*}};
\draw[-,dashed] (0,2.75) -- (3,2.75);
\draw[-,dashed] (3,0) -- (3,2.75);
\draw [-,thick] plot [smooth, tension=0.6] coordinates { (0,1) (0.5,1.1) (1,0.7) (2,1.4) (3,0.5) (4,1.1)};
\draw[-,dashed] (0,1.4) -- (2,1.4);
\draw[-,dashed] (2,0) -- (2,1.4);
\draw [decorate,
    decoration = {brace}, thick, gray] (-0.1,1.4) --node[left]{\scriptsize $O(1)$} (-0.1,2.75);
\node at (2,-1) {\small Approximation (Cor.~\ref{cor:robustness})};
\node at (2,-1.5) {\small submodular};
\draw[-,thick] (-0.1,0.2) -- (0.1,0.3);
\draw[-,thick] (-0.1,0.3) -- (0.1,0.4);
\end{tikzpicture}
\end{minipage}
\end{center}
\caption{Visualization of our results for submodular and XOS rewards. Lifting (left panel): We are given as input a contract $\contract^\star$ and a CCE (triangle) and we construct a contract $\contract$ and a PNE (circle). Robustness (middle panel): We are given a contract $\contract^\star$ and a PNE (circle) and we construct a contract $\contract$ with a guarantee for any CCE under that contract. 
In both panels, the ranges of principal utilities under PNE and CCE are drawn as continuous intervals for illustration.
Approximation (right panel): The lower curve shows the worst-performing CCE for each contract; the upper curve shows the corresponding best-performing CCE. We show that it is possible to compute a contract $\contract$ such that the worst CCE under that contract is within a constant-factor of the best CCE under any contract (black triangle).} 
\label{fig:gaps}
\end{figure}

\subsubsection{Main Results: Submodular and XOS Rewards}
\label{sec:results-main}

For submodular rewards, \cite{BabaioffFN10} demonstrated 
that the principal's utility under the best mixed Nash equilibrium can be strictly greater than under the best pure Nash equilibrium.
They conjectured that, for submodular rewards, this (multiplicative) gap is bounded by a constant.
Despite substantial progress on combinatorial contract design since then, this conjecture has remained unresolved.

Our main result not only settles this conjecture but in fact establishes a much stronger result.
We show that, in a model with arbitrary (non-binary) actions and  for XOS rewards (a strict generalization of submodular rewards), for each contract $\contract^\star$ and any coarse-correlated equilibrium (CCE) of that contract, 
there is a (typically different) contract $\contract$ and a pure Nash equilibrium (PNE) under $\contract$ that achieves a constant fraction of the principal's utility under the CCE. See Figure~\ref{fig:gaps} (left) for an illustration.
Thus, the conjectured constant gap is strengthened along three dimensions: extending from binary to arbitrary actions, from MNE to CCE, and from submodular to XOS rewards.

\medskip
\noindent\textbf{Black-Box Lifting Theorem} 
(Theorem~\ref{thm:constant-gap})\textbf{.} 
For the multi-agent combinatorial contracts 
model with submodular or XOS rewards, given any contract $\contract^\star$ and any CCE of $\contract^\star$, there exists a contract $\contract$ and a PNE $S$ of $\contract$ such that the principal's utility under $S$ is a constant fraction of the principal's utility under the CCE of $\contract^\star$. Moreover, given $\contract^\star$ and the CCE, one can find $\contract$ and $S$ in polynomial time (in $n$, $m$ and the support size of the CCE) with value and demand oracle access to $f$. 
\medskip

This result provides a black-box lifting: 
any guarantee established with respect to the best-PNE benchmark automatically holds also against the stronger benchmark of the best CCE, 
with only a constant factor loss, 
thereby strengthening existing bounds for multi-agent contracts and simplifying the path to new ones.
Indeed, to obtain guarantees against the best CCE, one no longer needs to grapple with the intricate notion of CCE (with its randomization and correlation), but can instead reason about the simpler and more transparent concept of pure Nash equilibrium. Readers familiar with the smoothness framework for bounding the price of anarchy \citep{Roughgarden15} may find this connection conceptually reminiscent; see Section \ref{sec:related-work} for further discussion.

As an immediate corollary, all approximation results established in \cite{DEFK23,DEFK25}, for instance, hold against the stronger benchmark of the best CCE under any contract. 
One such result is that, for submodular rewards, there is a poly-time algorithm that finds a contract under which any PNE attains a constant-factor approximation to the best PNE under any contract~\citep{DEFK25}. Our lifting theorem immediately extends this guarantee against the best CCE under any contract.

% We further strengthen this result by providing a black-box robustness theorem 
% for submodular rewards that transforms any contract and pure Nash equilibrium of that contract into a different contract that comes with a guarantee for any CCE under that contract. See Figure~\ref{fig:gaps} (middle) for an illustration.

We further strengthen this result by establishing the following (black-box) robustness result 
for submodular rewards, which shows how to transform any contract and pure Nash equilibrium of that contract into a different contract that comes with a guarantee for any CCE under the modified contract. See Figure~\ref{fig:gaps} (middle) for an illustration.

\medskip
\noindent\textbf{Black-Box Robustness Theorem} (Theorem~\ref{thm:constant-gap-sub-CCE})\textbf{.} For the multi-agent combinatorial contracts model with submodular rewards, given any contract $\contract^\star$ and pure Nash equilibrium of $\contract^\star$, there is an algorithm that runs in polynomial time (in $n$ and $m$) using value queries to $f$, that computes a contract $\contract$ such that any CCE  under $\contract$ achieves an $O(1)$-approximation to the principal's utility under the PNE of $\contract^\star$.
\medskip

Combining this with the results in \citep{DEFK23,DEFK25} and the Black-Box Lifting Theorem, shows that for the multi-agent combinatorial contracts 
model with submodular rewards there is a poly-time algorithm (with value and demand oracles) that computes a contract 
such that any CCE under that contract provides a $O(1)$-approximation to the best CCE under any contract (see Corollary~\ref{cor:robustness}). In the special case of binary actions, such a contract can be computed with value queries only. We provide an illustration of the guarantees achieved by the corresponding algorithms in Figure~\ref{fig:gaps} (right).

Similar to the Black-Box Lifting Theorem, the Black-Box Robustness Theorem is conceptually related to the price of anarchy framework. It shows that, for submodular rewards, one only needs to show  the existence of a good contract and pure Nash equilibrium, and can then rely on the extension result to transform the original contract into a contract under which all CCE (and hence learning outcomes) are near-optimal.
Notably, for XOS rewards, such a worst-case approximation guarantee is unattainable, even for the weaker target of ensuring that 
any PNE under the given contract is within a constant factor of the best PNE under any contract (as already observed in \cite{DEFK25}).

\subsubsection{Beyond XOS Rewards: Mapping the Landscape}
\label{sec:results-additional}

\begin{table}[t]
    \centering
    \renewcommand*{\arraystretch}{1.1}
    \begin{tabular}{|c|c|p{1.75cm}|p{1.75cm}|p{1.75cm}|}
        \hline
        \multirow{2}{*}{Reward function} & & \multicolumn{3}{c|}{Gap between ... and PNE} \\
        \cline{3-5}
        & & ~~~~MNE & ~~~~CE & ~~~~CCE \\
        \hline
        \hline
        \multirow{2}{*}{Submodular/XOS} & \multirow{1}{*}{Binary/arbitrary} &  \multicolumn{3}{c|}{$\Theta(1)$} \\
        & \multirow{1}{*}{actions} & \multicolumn{3}{c|}{Lower: Ex.~\ref{ex:mixed}, Upper: Thm.~\ref{thm:constant-gap}} \\
        \hline
        \hline
        \multirow{2}{*}{Subadditive} & \multirow{1}{*}{Binary/arbitrary} &  \multicolumn{3}{c|}{$\Theta(\text{poly}(n))$} \\
        & \multirow{1}{*}{actions} & \multicolumn{3}{c|}{Lower: Prop.~\ref{prop:subadditive-binary}, Upper: Prop.~\ref{prop:subadditive-upper-bound}} \\
        \hline
        \hline
        \multirow{4}{*}{Supermodular} & \multirow{1}{*}{Binary} &   \multicolumn{3}{c|}{No gap} \\
        & \multirow{1}{*}{actions} & \multicolumn{3}{c|}{Thm.~\ref{thm:no-gap-super-binary}} \\
        \cline{2-5}
        & \multirow{1}{*}{Arbitrary} & \multicolumn{2}{c|}{No gap} & Unbounded \\
        & \multirow{1}{*}{actions} & \multicolumn{2}{c|}{Thm.~\ref{thm:no-gap-super-multi}} & Prop.~\ref{prop:supermodular-gap-for-cce}\\
        \hline
        \hline
        \multirow{2}{*}{General} & \multirow{1}{*}{Binary/arbitrary} &   \multicolumn{3}{c|}{Unbounded} \\
        & \multirow{1}{*}{actions} &  \multicolumn{3}{c|}{Prop.~\ref{prop:general-binary}} \\
        \hline
    \end{tabular}
    \caption{Gaps between MNE/CE/CCE and PNE, for different reward functions and binary actions vs.~arbitrary actions.
    In merged cells, lower bounds are proved for simpler action and equilibrium notions and extend to more general ones; upper bounds are proved for more general cases and apply to the simpler ones. (E.g., for submodular/XOS rewards, the lower bound is proved for binary actions and MNE, while the upper bound is proved for arbitrary actions and CCE.) 
    }
    \label{tab:equilibrium_results}
\end{table}

We next examine the gap between different equilibrium concepts beyond XOS rewards. We start with subadditive rewards, then we move to supermodular rewards, and finally we consider general (monotone) rewards.

\paragraph{Subadditive Rewards.}
For subadditive rewards, we construct a carefully designed instance that exhibits a polynomial lower bound of order $\Omega(\sqrt{n})$ on the gap between the principal's utility under the best mixed Nash equilibrium (MNE) and the best pure Nash equilibrium (PNE), where $n$ denotes the number of agents.
Remarkably, this gap already arises in the binary-actions case (so $n$ is also the number of actions here). We also show that this gap is at most $O(n)$, even between CCE and PNE and the case of general actions.\footnote{A weaker upper bound of $m$ follows from a recent result by \cite{DEFK25}, showing that the gap between social welfare and the best PNE is at most $m$.}

\medskip
\noindent\textbf{Proposition (Subadditive Rewards)} (Proposition~\ref{prop:subadditive-binary} and Proposition~\ref{prop:subadditive-upper-bound})\textbf{.} 
For the multi-agent combinatorial contracts %action 
model with subadditive rewards,  the gap between the principal's utility from the best MNE and the best PNE is $\Omega(\text{poly}{(n)})$, even in the binary-actions case. On the other hand,
the gap between the principal's utility under the best CCE and the best PNE is at most $O(n)$, 
even for an arbitrary number of actions.
%\medskip

\paragraph{Supermodular Rewards.} 
For supermodular rewards, 
we show that in the case of \emph{binary actions}, there is no gap between the principal's utility from the best coarse-correlated equilibrium and the best pure Nash equilibrium. 
This strengthens a result of \cite{BabaioffFN10}, who established this no-gap result only with respect to mixed NE. 

\medskip
\noindent \textbf{Theorem (Supermodular Rewards, Binary Actions)} (Theorem~\ref{thm:no-gap-super-binary})\textbf{.} For the multi-agent model with binary actions and supermodular rewards, there is no gap between the principal's utility from the best CCE and the best PNE.
\medskip

In contrast, for arbitrary actions there is no gap only when comparing the best \emph{correlated} equilibrium (CE) to the best PNE; for the more general notion of \emph{coarse-correlated} equilibria (CCE), the gap to PNE may be unbounded.

\medskip
\noindent \textbf{Theorem (Supermodular Rewards, General Actions)} (Theorem~\ref{thm:no-gap-super-multi} and Proposition~\ref{prop:supermodular-gap-for-cce})\textbf{.} For the multi-agent model with arbitrary actions and supermodular rewards, there is no gap between the principal's utility under the best CE and the best PNE; while the gap between CCE and PNE is unbounded.

\paragraph{General Rewards.} Finally, we turn to general monotone rewards, in particular rewards that are neither subadditive nor supermodular. 
We show that in this regime, even the gap between the best MNE and the best PNE can be unbounded.
Moreover, such unbounded gaps arise already with a constant number of agents and binary actions.

\medskip
\noindent\textbf{Proposition (General Rewards)} (Proposition~\ref{prop:general-binary})\textbf{.}
In the multi-agent combinatorial contract model with general rewards (neither subadditive, nor supermodular), the gap between the principal's utility under the best MNE and the best PNE is unbounded,
even with only four agents and binary actions.

\subsection{Challenges and Techniques}

\paragraph{Black-Box Lifting and the Scaling-for-Existence Lemma.} 
Our key tool for establishing an upper bound on the gap between coarse-correlated equilibria and pure Nash equilibria, for submodular and XOS rewards, is a 
simple yet powerful  
\emph{Scaling-for-Existence Lemma} (Lemma~\ref{lem:double}). This lemma starts from a contract $\contract^\star$ and a distribution over action sets that satisfies a mild dropout-stability 
condition
(Definition~\ref{def:dropout-stability}). Dropout-stability requires that each agent weakly prefers to follow the actions in the given distribution over unilaterally deviating to taking no action. This condition is naturally satisfied by all equilibrium concepts we study. %(PNE, MNE, CE, CCE).

The Scaling-for-Existence Lemma 
shows that, 
after appropriately scaling the contract to $\contract$, there
exists a PNE
that guarantees high reward relative to the original distribution.
In particular, if there is a subset of agents that achieves high  
expected reward under the original distribution 
and whose total share $\sum_i \alpha^\star_i$ 
is bounded away from $1$, 
then the lemma can be used to show that the induced pure Nash equilibrium recovers a constant fraction of the principal's utility under the original distribution.
To obtain the constant-factor gap, we show that in any dropout-stable distribution, either (i) there is a ``significant'' agent and a good pure Nash equilibrium that incentivizes only this agent, or (ii) there exists a subset of agents to which we can apply the Scaling-for-Existence Lemma to obtain a good pure Nash equilibrium.

Our proof of the Scaling-for-Existence Lemma leverages that, for each contract, the 
induced game among the agents is a potential game (e.g., \cite{VDPP24,DEFK25}). 
We show that 
for the original contract $\contract^\star$ and any dropout-stable distribution under this contract, the XOS structure implies
that the expected potential of any set $S$ in the support of the dropout-stable distribution
is non-negative. 
We then scale the contract to
$\contract$ and select $S$ as a global maximizer of the potential function for $\contract$. This ensures that $(S,\contract)$ forms a pure Nash equilibrium. 
Finally, the fact that $\contract$ is obtained from $\contract^\star$ through scaling, together with the non-negativity of the expected potential 
under $\contract^\star$, implies that this pure Nash equilibrium achieves high reward relative to the original distribution.

Notably, for subadditive rewards, our proof for the Scaling-for-Existence Lemma breaks: it is no longer guaranteed that the expected potential of any set $S$ in the support of a dropout stable distribution for contract 
$\contract^\star$ is non-negative. 
This failure not only undermines the proof technique, but points to a deeper structural issue. Namely, such a result is provably impossible for subadditive rewards: the gap between PNE and MNE (and thus PNE and CCE) is superconstant.

\paragraph{Black-Box Robustness and the Scaling-for-Robustness Lemma.} The Scaling-for-Existence Lemma 
derives a good PNE from 
a dropout-stable distribution over actions (such as CE or CCE). Our Black-Box Robustness Theorem builds on a lemma, the \emph{Scaling-for-Robustness Lemma} (Lemma~\ref{lem:scaling-for-robustness}), that achieves an orthogonal goal, namely showing that \emph{every} CCE is good relative to a reference equilibrium (or rather relaxation thereof).

Our lemma generalizes the Doubling Lemma of \cite{DEFK25}. The Doubling Lemma showed that, for submodular rewards, for every contract $\contract^\star$ and PNE of $\contract^\star$, there exists a scaled contract $\contract$ such that any PNE of $\contract$ provides a constant-approximation to the original PNE.
We present a new argument that achieves a parallel result for any CCE (rather than PNE) under the scaled
contract, losing only another constant factor.
This enables a comparison of any CCE under a given contract to the best PNE under any contract.

There are two main differences between the two scaling lemmas: 
First, they give different types of guarantees: the Scaling-for-Existence Lemma ensures that there {\em exists} a good equilibrium, whereas the Scaling-for-Robustness Lemma ensures that {\em every} equilibrium is good. Second, their domains differ: the existence lemma holds for XOS rewards, while the robustness lemma holds only for submodular rewards. 
A lemma unifying both results provably can't hold. Indeed, for XOS rewards, there are instances for which the gap between the best and the worst equilibrium is large \citep{DEFK25}.

\paragraph{Transitions for Supermodular Rewards.}
Our results for supermodular rewards exhibit two perhaps surprising transitions, namely (i) between binary and general actions and (ii) between CE and CCE. 
Compared to the binary-actions case, arbitrary actions introduce two additional challenges.
First, agents may take actions not present in the original distribution. We address this by showing
that there exists a PNE whose action profile contains the union of the sets in the support of
the correlated equilibrium, rather than matching it exactly, as in the binary case.
Second, a distribution that is only a CCE (and not a CE) may place positive probability on
actions that reduce an agent’s utility under the contract; such actions are not best responses
and therefore cannot appear in any PNE of that contract. Consequently, our proof extends to
correlated equilibria but not to coarse-correlated equilibria. Indeed, our lower-bound construction
for CCE shows that the gap can be unbounded.

\subsection{Further Related Work}
\label{sec:related-work}
\paragraph{Inefficiency of Equilibria and Price of Anarchy.} 
Bounding the inefficiency of equilibria is one of the staples of algorithmic game theory \cite[e.g.,][]{KoutsoupiasP99,RoughgardenT02,SyrgkanisT13}, and among the main conceptual contributions of theoretical computer science to economics.

In this regard, our work bears some conceptual resemblance to the smoothness framework for bounding the price of anarchy \citep{Roughgarden15}, in that to bound a ratio under a general equilibrium concept (CCE), it suffices to bound it under a simpler one (PNE), thereby simplifying the analysis. Similarly, as in the smoothness framework, it suffices to consider ``simple" deviations: deviation to one's strategy in an optimal profile in the smoothness framework, and deviation to taking no action in our case (as captured by the dropout-stability notion).

Two important differences to the Price of Anarchy literature are that the Price of Anarchy literature is typically interested in the social welfare (or social cost) achieved by equilibria, and focuses on worst-case equilibria for a fixed game. In contrast, here we are interested in the principal's utility and we compare equilibria across different induced games.

\paragraph{Related Studies on Combinatorial Contracts.}
Our work builds on the growing body of research on combinatorial contracts~\citep{Feldman25}, and in particular contributes to the rapidly expanding literature on multi-agent contracts.
The work that is most related to our work is \cite{BabaioffFN10}, who---working in the combinatorial agency model of \cite{BFN06,BabaioffFFW2012}---quantify the ratio between the principal's utility under mixed and pure strategies, for reward functions exhibiting increasing and decreasing marginals (corresponding to supermodular and submodular functions, respectively). 
We significantly extend this work from binary actions to arbitrary actions, and from mixed Nash equilibria to more general equilibrium concepts such as correlated and coarse-correlated equilibria---and on the way we resolve the main open question from this prior work, concerning the gap between PNE and MNE for submodular rewards

The very recent work of \cite{DEFK25} is another important reference. Although their work focuses on pure Nash equilibria, just like most of the existing literature on multi-agent contracts, they do show a related result: namely, that for the multi-agent combinatorial actions model that we study in this work, and for submodular rewards, there exists a contract $\contract$ such that the worst-equilibrium under that contract, yields a constant-factor approximation to the best equilibrium under the best contract. We strengthen this result in two ways: First, we show that in fact any CCE of the contract returned by the algorithm obtains a constant fraction to the best PNE under any contract. Second, by our lifting result, we extend this guarantee to the best-CCE benchmark.

\cite{alon2025multi} extend the multi-agent (binary action) contract design framework to settings with many projects, where the principal needs to partition the agents among the projects, and within each project, the principal incentivizes the agents through a contract.

\cite{CacciamaniBCN24} explore a related multi-agent and (non-combinatorial) multi-action model, and highlight the value of randomized contracts. In their model, the principal also uses a randomized contract, which is different from our approach which assumes that the principal posts a single deterministic contract. They then propose and study an equilibrium notion, which can be interpreted as a correlated equilibrium when both the principal and the agents randomize. They show that this type of randomized correlated equilibria are more powerful than pure Nash equilibria, demonstrating for instance that the principal's utility can be unboundedly higher when the setting has a supermodular 
reward structure. 
An important difference between our work and their work is that we consider deterministic contracts.

Finally, slightly more removed from our question and work, \cite{DasarathaGS2024} study a multi-agent multi-action model in which the agents choose from a continuum of effort levels, and the agent's choices determine both their costs and the principal's reward, through cost and reward functions. 
The conceptual similarity lies in their treatment of ``fractional'' actions, akin to our model where an agent selects different action sets with varying probabilities. However, unlike our work, they do not consider mixed or correlated equilibria.

%% PRELIMINARIES
\section{Preliminaries}
\label{sec:model}

\newcommand{\CCE}[0]{\mathcal{D}}
\newcommand{\dist}[0]{\mathcal{D}}
\newcommand{\actions}[0]{A}

\paragraph{Multi-Agent Combinatorial Contracts.} Every agent $i \in [n]$ has a (finite) set of actions $\actions_i$. The set of actions $A_i$ and $A_{i'}$ of any two agents $i \neq i'$ are disjoint. We denote the set of all actions by $\actions = \bigcup_{i \in [n]} \actions_i = [m]$. An important special case is when for all $i \in [n]$, it holds that $|\actions_i| = 1$. In this case, each agent can either take action or not. We refer to this as the binary-actions case.

Every agent $i$ can take any subset of actions $S_i \subseteq \actions_i$. 
We use $S = \bigcup_{i \in [n]} S_i$ to denote the set of actions chosen by the agents, and for each agent $i \in [n]$ we let $S_{-i} = \bigcup_{i' \neq i} S_i.$ (Similarly, for any $S \subseteq \actions$ we let $S_i = S \cap \actions_i$.)
There is a reward function $f: 2^A \rightarrow \nnreals.$ We generally assume that the reward function is monotone (non-decreasing) and normalized so that $f(\emptyset) = 0$. In addition, each agent has a cost $c_j$ for each action $j \in \actions_i$. The cost of a set of actions $S_i \subseteq \actions_i$ is $c(S_i) = \sum_{j \in S_i} c_j.$ For singletons $\{j\} \subseteq \actions$, we sometimes use the shorthands $f(j) = f(\{j\})$ and $c(j) = c(\{j\})$. 

A (linear) contract $\contract = (\alpha_1, \ldots, \alpha_n)$ defines a share $\alpha_i \in [0,1]$ of the reward $f(S)$ that is to be paid to agent $i$.  The utility of agent $i$ under contract $\contract$ when the set of actions chosen by the agents is $S$ is
\[
u_i(S, \contract) = \alpha_i f(S) - c(S_i).
\]

The utility of the principal in that case is
\[
u_P(S, \contract) = \left(1-\sum_{i \in [n]} \alpha_i\right) \cdot f(S).
\]

Note that, for any set $S$, the utilities of the agents and the utility of the principal sum up to $f(S) - c(S)$. We refer to this quantity as the \emph{welfare} of the set of actions $S$.

We assume that both the principal and the agents are expected utility maximizers; 
their expected utilities for a distribution $\CCE$ over sets $S \subseteq \actions$ is $u_i(\CCE,\contract) = \E_{S \sim \CCE}[u_i(S,\contract)]$ and $u_P(\CCE,\contract) = \E_{S \sim \CCE}[u_P(S,\contract)]$. We adopt the perspective of the principal and seek a contract that maximizes the principal's expected utility.

\paragraph{Equilibrium Concepts.} In multi-agent combinatorial contracts, the principal defines a contract $\contract$ and the agents choose their actions in reply to this contract. 
We are thus looking at a single-leader multiple-followers Stackelberg game. Specifically, for any fixed contract $\contract$ we study a simultaneous move game among the agents. We consider different equilibrium concepts.

A pure Nash equilibrium is a set of actions $S$ such that no agent wants to deviate from their choice of action $S_i$ to some other action $T_i$.

\begin{definition}[Pure Nash Equilibrium]
Given reward function $f$, a \emph{pure Nash equilibrium} (PNE) of contract $\contract$ is a set of actions $S \subseteq \actions$ such that for every agent $i$ and every $T_i \subseteq \actions_i$ it holds that $$\alpha_i \cdot f(S_{-i} \cup S_i) -c(S_i) \geq \alpha_i \cdot f(S_{-i} \cup T_i) -c(T_i) .$$
\end{definition}

A mixed Nash equilibrium, in turn, is a product distribution $\dist = \dist_1 \times \ldots \times \dist_n$ over sets of actions where $\dist_i$ is a distribution over sets of action among $A_i$, where no agent $i$ strictly wants to deviate from their distribution $\dist_i$ to a pure strategy $T_i \subseteq A_i$.

\begin{definition}[Mixed Nash Equilibrium]
Given reward function $f$, a \emph{mixed Nash equilibrium} (MNE) of contract $\contract$ is a product distribution $\dist = \prod_{i \in [n]} \dist_i$ over sets $S \subseteq \actions$ such that for every agent $i$ and every $T_i \subseteq A_i$ it holds that $$\E_{S \sim \dist}[\alpha_i \cdot f(S_{-i} \cup S_i) -c(S_i)] \geq \E_{S \sim \dist}[\alpha_i \cdot f(S_{-i} \cup T_i) -c(T_i)]  .$$
\end{definition}

Correlated equilibria are defined as a joint (possibly) correlated distribution over sets of actions. We can interpret $S_i$ as the action recommended to agent $i$. Then the equilibrium requirement is that each agent $i$ should prefer to follow their recommended action $S_i$, rather then switching to some other action $\pi_i(S_i)$ whenever they are told $S_i$.

\begin{definition}[Correlated Equilibrium]
Given reward function $f$, a \emph{correlated equilibrium} (CE) of contract $\contract$ is a distribution $\dist$ over sets $S \subseteq \actions$ such that for every agent $i$ and every mapping $\pi_i: 2^{\actions_i} \rightarrow 2^{\actions_i}$ it holds that $$\E_{S \sim \dist}[\alpha_i \cdot f(S_{-i} \cup S_i) -c(S_i)] \geq \E_{S \sim \dist}[\alpha_i \cdot f(S_{-i} \cup \pi_i(S_i)) -c(\pi_i(S_i))]  .$$
\end{definition}

The notion of a coarse-correlated equilibrium weakens the equilibrium requirement of a correlated equilibrium. Rather than being able to map each recommended action $S_i$ to some other action $\pi_i(S_i)$, agent $i$ only considers a single alternative action $T_i$.

\begin{definition}[Coarse-Correlated Equilibrium]
Given reward function $f$, a \emph{coarse-correlated equilibrium} (CCE) of contract $\contract$ is a distribution $\dist$ over sets $S \subseteq \actions$ such that for every agent $i$ and every $T_i \subseteq A_i$ it holds that $$\E_{S \sim \dist}[\alpha_i \cdot f(S_{-i} \cup S_i) -c(S_i)] \geq \E_{S \sim \dist}[\alpha_i \cdot f(S_{-i} \cup T_i) -c(T_i)]  .$$
\end{definition}

Note that the inequality in the definition of CCE is identical to that in the definition of MNE. The difference is that 
in an MNE, the distribution 
$\dist$ must be a product distribution.

It is not difficult to see that the equilibrium notions are successive relaxations, i.e.,
\[
\text{PNE} \subseteq \text{MNE} \subseteq \text{CE} \subseteq \text{CCE}.
\]

While in general games, a pure Nash equilibrium may not exist, games induced by multi-agent combinatorial action contracts 
admit a  
weighted potential function \citep{VDPP24,DEFK25}. 
Recall that any (finite) game that admits such a potential function has the finite-improvement property, and hence at least one pure Nash equilibrium. Another useful property of such games is that any local maximum of the potential function corresponds to a pure Nash equilibrium of the underlying game. (For a detailed discussion of potential games see \cite{MS96}.)

Recall that a %(generalized, ordinal) 
weighted potential function is a function $\phi: 2^\actions \rightarrow \{-\infty\} \cup \reals$ such that $u_i(S_i,S_{-i},\contract) > u_i(S'_i,S_{-i},\contract)$ implies $\phi(S_i,S_{-i}) > \phi(S'_i,S_{-i})$. For any set $S$ and contract $\contract$, let 
\begin{equation}
\Phi(S, \contract) = f(S) - \sum_{i \in [n]} \frac{c(S_i)}{\alpha_i},
\label{eq:potential}    
\end{equation}
where, if $\alpha_i=0$ we define $c(S_i)/\alpha_i=\infty$ when $c(S_i)>0$, and we let $c(S_i)/\alpha_i=0$ when $c(S_i)=0$.

\begin{proposition}[\cite{VDPP24,DEFK25}]
\label{prop:potential}
For every contract $\contract$, the function $\Phi(\cdot,\contract)$ 
is a weighted 
potential function in the game induced by the contract $\contract$.
\end{proposition}

\paragraph{Dropout Stability.}

A crucial ingredient in our analysis is the following property, which captures stability with respect to a unilateral deviation to taking no action.

\begin{definition}[Dropout Stability]
\label{def:dropout-stability}
Given reward function $f$, a distribution $\dist$ over sets of actions $S$ is \emph{dropout-stable} with respect to $\contract$ if for every agent $i$ it holds that  $$\E_{S \sim \dist}[\alpha_i \cdot f(S) -c(S_i)] \geq \E_{S \sim \dist}[\alpha_i \cdot f(S_{-i})].$$
\end{definition}

The following observation is immediate.

\begin{observation}
\label{obs:cce-stable}
Any coarse-correlated equilibrium of $\contract$ is also dropout-stable with respect to $\contract$.
\end{observation}

In what follows, it will be useful to have the notion of marginal contribution of an individual action or set of actions to a given set of actions $S$. For an action $j \in \actions$, we write $f(j \mid S) = f(S \cup j) - f(S)$ for the marginal contribution of action $j$ to $S$. Similarly, for a set of actions $T \subseteq\actions$, we write $f(T \mid S) = f(S \cup T) - f(S)$ for the marginal contribution of the set of actions $T$ to $S$.

\paragraph{Classes of Reward Functions.} 
Our main interest will be in the following classes of (non-negative) reward functions $f: 2^A \rightarrow \nnreals$. A set function is: 
\begin{itemize}[leftmargin=*, itemsep=0pt, parsep=2pt, topsep=8pt]
\item 
\emph{additive} if there exist values $f_1, \ldots, f_m \in \nnreals$ such that  $f(S) = \sum_{j \in S} f_j$.
\item  
\emph{gross-substitutes} if for any two vectors $p\leq  q \in \reals_+^m$ and any $S \subseteq \actions$ such that $S \in \arg\max_{S' \subseteq \actions} (f(S') - \sum_{j \in S'} p_j)$ there is a $T \subseteq \actions$ with $\{j \in S \mid q_j \leq p_j\} \subseteq T$ such that $T \in \arg\max_{T' \subseteq \actions} (f(T') - \sum_{j \in T'} q_j)$.
\item 
\emph{submodular} if for every $S, S' \subseteq \actions$ with $S \subseteq S'$ and any $j \in \actions$ it holds that $f(j \mid S) \geq f(j \mid S')$.
\item 
\emph{XOS} if there exists a collection of additive functions  $\{a_\ell: 2^\actions \rightarrow \nnreals\}_{\ell = 1,\ldots, k}$ such that for each $S \subseteq \actions$ it holds that $f(S) = \max_{\ell=1,\ldots,k}  a_{\ell }(S)$. 
\item 
\emph{subadditive} if for every $S,S' \subseteq \actions$, it holds that $f(S) + f(S') \geq f(S \cup S').$
\item 
\emph{supermodular} if for every $S, S' \subseteq \actions$ with $S \subseteq S'$ and any $j \in \actions$ it holds that $f(j \mid S) \leq f(j \mid S')$.
\end{itemize}

All functions in this list except for supermodular belong to the hierarchy of complement-free set functions. It is well known that additive $\subseteq$ gross substitutes $\subseteq$ submodular $\subseteq$ XOS $\subseteq$ subadditive, and that all containment relations are strict \citep{LLN06}.

We consider two standard primitives for accessing combinatorial set functions. A \emph{value oracle} is given a set $S$ and returns $f(S)$. A \emph{demand oracle} is given a set of non-negative prices $p_1, \ldots, p_m \in \mathbb{R}_{\geq 0}$ and returns a set $S$ that maximizes $f(S) - \sum_{j \in S} p_j$.

%% SECTION: SUBMODULAR REWARDS
\section{Submodular and XOS Rewards}

In this section, we show that, when the rewards are submodular or more generally XOS,
% then
there is a constant gap between the principal's utility under the best coarse-correlated equilibrium and the best pure Nash equilibrium. %\pdc{updated the theorem statement to use $\dist^\star$ for CCE}
This result is tight in several ways. First, as we demonstrate in Example~\ref{ex:mixed}, there is at least a constant gap between mixed Nash equilibria and pure Nash equilibria, even with %submodular
gross-substitutes rewards and binary actions. Moreover, as we show in Proposition~\ref{prop:subadditive-binary}, for subadditive rewards there is at least a polynomial gap, even between mixed and pure Nash equilibria and the special case of binary actions.

\begin{theorem}[Black-Box Lifting Theorem]\label{thm:constant-gap}
Suppose $f$ is XOS. Let $\contract^\star$ be any contract and let $\dist^\star$ be any coarse-correlated equilibrium of $\contract^\star$. Then there exists a contract $\contract$ and a pure Nash equilibrium $S$ of $\contract$ such that $(1-\sum_{i} \alpha_i) \cdot f(S) \geq \Omega(1) \cdot (1-\sum_{i} \alpha^\star_i) \cdot \mathbf{E}_{S^\star \sim \dist^\star}[f(S^\star)]$.
Moreover, given $\contract^\star$ and $\dist^\star$, one can find such $\contract$ and $S$ in polynomial time (in $n$, $m$ and in the support size of $\dist^\star$) using value and demand oracles to $f$.
\end{theorem}

To prove Theorem~\ref{thm:constant-gap} we proceed in two steps. In Section~\ref{sec:amplification-lemma}, we present our key new lemma driving this result, the Scaling-for-Existence Lemma. Afterwards, in Section~\ref{sec:proof-of-constant-gap}, we show how to use this lemma to establish the theorem. 
Finally, in Section~\ref{sec:robust-algo} we show how to leverage the theorem to obtain a poly-time algorithm for submodular rewards, that computes 
a contract such that any CCE under that contract provides a constant approximation to the best CCE under the best contract.
 
\begin{remark} 
We note that the proof of Theorem~\ref{thm:constant-gap} 
does not use the additivity of the cost functions, and thus Theorem~\ref{thm:constant-gap} holds even when each agent has an arbitrary normalized (not necessarily monotone) non-negative combinatorial cost function $c_i:2^{A_i} \rightarrow \reals_{\geq 0}$.
\end{remark}

\subsection{The Scaling-for-Existence Lemma}\label{sec:amplification-lemma}

Our key tool for establishing Theorem~\ref{thm:constant-gap} is the following Scaling-for-Existence Lemma. The starting point of this lemma is a coarse-correlated equilibrium $\dist$ for some contract $\contract$ --- or, more precisely, any distribution $\dist$ over sets that satisfies the weaker dropout-stability condition for contract $\contract$. The lemma then establishes the existence of a pure Nash equilibrium at an appropriately scaled contract $\contract'$, that achieves high reward relative to the original distribution $\dist$. Our proof relies on the property that the expected potential value of any dropout-stable distribution must be non-negative for XOS reward functions (even when restricted to a subset of agents). This implies that if a contract is scaled, then there is a PNE with high potential, which bounds from below the reward.

\begin{algorithm}[t]
\caption{Scaling-for-Existence for XOS Rewards\label{alg:lem}}
   \hspace*{\algorithmicindent} \textbf{Input:}  Costs $c_1,\ldots,c_m \in \reals_{\geq 0}$,  
   demand oracle access to a XOS function $f:2^A \rightarrow \reals_{\geq 0}$, a subset of agents $N'\subseteq N$, a parameter $\gamma>1$, a contract $\contract$ and a corresponding CCE $\dist$. \\
    \hspace*{\algorithmicindent} \textbf{Output:}  A contract $\contract'$ and a PNE $S'$ of $\contract'$ with  $f(S') \geq (1- \frac{1}{\gamma}) \cdot \E_{S\sim \dist}[f(\bigcup_{i \in N'} S_i)]$.
\begin{algorithmic}[1]
\State Let $\alpha_i' = \gamma\cdot \alpha_i\cdot \indicator{i\in N'}$
\State Let $\vec{p}$ be the price vector where $p(j)= \frac{c_j}{\alpha_i'}$ for each $j\in A_i$  
\Comment{Here, $\frac{0}{0}=0$  and $\frac{c}{0}=\infty$ for $c>0$} 
\State Let $S'\in\arg\max_{S\subseteq A} (f(S)-\sum_{j\in S}p(j))$
\State \Return $\contract',S'$
\end{algorithmic}
\end{algorithm}

\begin{lemma}
[Scaling-for-Existence Lemma]
\label{lem:double}
Suppose $f$ is XOS. Let $\dist$ be a dropout-stable distribution with respect to $\contract$. For any set of agents $N' \subseteq N$ and $\gamma > 1$, let $\contract'$ be defined by $\alpha'_i = \gamma \cdot \alpha_i$ for $i \in N'$ and $\alpha'_i = 0$ otherwise.
Then, there exists a pure Nash equilibrium $S'$ with respect to $\contract'$, satisfying  $$f(S') \geq (1- \frac{1}{\gamma}) \cdot \E_{S\sim \dist}[f(\bigcup_{i \in N'} S_i)].$$
Moreover, $\contract'$ and $S'$ can be found in polynomial time (in $n$, $m$ and in the support size of $\dist^\star$) with demand oracle access to $f$ (see Algorithm~\ref{alg:lem}).
\end{lemma}

\begin{proof}
    Consider any $N'$. We use notation $S_{N'} = \bigcup_{i \in N'} S_i$.
    Fix a set $S$ in the support of $\dist$. Because $f$ is XOS, there is an additive function $a$ such that $f(S) = a(S)$ and $f(S'') \geq a(S'')$ for all $S''$. Therefore, we have 
    \begin{equation}
    \label{eq:ineq1}
    \sum_{i \in N'} f(S_i \mid S_{-i}) = \sum_{i \in N'} \left(f(S) - f(S_{-i}) \right)\leq \sum_{i \in N'} \left(a(S) - a(S_{-i})\right) = \sum_{i \in N'} a(S_i) = \sum_{i \in N'} a(S_{N'}) \leq f(S_{N'}).
    \end{equation}

   Now for any set $S$ define $\Phi(S, \contract) = f(S) - \sum_i \frac{c(S_i)}{\alpha_i}$ where $\frac{0}{0}$ is interpreted as $0$, and $\frac{c}{0}$ for a positive $c$ is interpreted as $\infty$. 
   We observe that for any agent $i\in N$ with $\alpha_i=0$, by dropout-stability it holds that $\E_{S\sim \dist }[c(S_i)] =0 $ and therefore $\Phi(S,\contract)$ is always finite for $S$ in the support of a distribution $\dist$ that is dropout-stable with respect to $\contract$.

    Then we have
    \[
    \Phi(S_{N'}, \contract)  = f(S_{N'}) - \sum_{i \in N'} \frac{c(S_i)}{\alpha_i} \geq \sum_{i \in N'} f(S_i \mid S_{-i}) - \sum_{i \in N'} \frac{c(S_i)}{\alpha_i} = \sum_{i \in N'} \left( f(S_i \mid S_{-i}) - \frac{c(S_i)}{\alpha_i} \right),
    \]
    where the inequality follows by Inequality~\eqref{eq:ineq1}.

    Observe that dropout-stability is equivalent to $\E_{S \sim \dist}[\alpha_i \cdot f(S_i \mid S_{-i}) -c(S_i)] \geq 0$ for all $i$.
    Then, by linearity of expectation and dropout-stability, it follows that 
    \begin{equation}
        \label{eq:phi-positive}
        \E_{S \sim \dist}[\Phi(S_{N'}, \contract)] \geq 0.
    \end{equation}

    Let $S'$ be a set of actions maximizing  $\Phi(S', \contract')$ (when fixing  $\contract'$ as defined in the statement of the lemma). Then we have  
    \begin{align*}
    \Phi(S', \contract') &\geq \E_{S\sim \dist}[\Phi(S_{N'}, \gamma \contract)] = \E_{S\sim \dist }\left[f(S_{N'}) - \sum_{i \in N'} \frac{c(S_i)}{\gamma \contract_i}\right]\\
    &= \left(1-\frac{1}{\gamma}\right) \E_{S \sim \dist}[f(S_{N'})] + \frac{1}{\gamma} \E_{S\sim \dist}[\Phi(S_{N'},\contract)] \geq \left(1-\frac{1}{\gamma}\right) \E_{S\sim \dist}[f(S_{N'})],
    %\frac{1}{2} \Ex{f(S_{N'}) + \Phi(S_{N'}, \contract)} \geq \frac{1}{2} \Ex{f(S_{N'})}.
    \end{align*}
    where the first inequality follows by the maximality of $S'$, and the last inequality follows by Eq.~\eqref{eq:phi-positive}.

    As $S'$ is a global maximum of $\Phi(\cdot, \contract')$, it is also a local maximum. 
    Since $\Phi(\cdot,\contract')$ is a potential function for the game induced by the contract $\contract'$ (see Proposition~\ref{prop:potential}), this means that $S'$ is a pure Nash equilibrium with respect to contract $\contract'$. 
    Finally, observe that one can find $S'$ using a single demand query with prices $\frac{c_j}{\alpha'_i}$ for each action $j\in A_i$ (where $\frac{0}{0}$ is interpreted as $0$, and $\frac{c}{0}$ for $c>0$ is interpreted as $\infty$).
    The ``moreover'' part of the lemma thus follows by noting that Algorithm~\ref{alg:lem} sets $\contract'$ as stated in the lemma, and chooses $S'$ as a demand set at these prices.
\end{proof}

\subsection{Proof of Theorem~\ref{thm:constant-gap}}\label{sec:proof-of-constant-gap}

We are now ready to prove Theorem~\ref{thm:constant-gap}. The high-level idea is to 
distinguish cases 
based on the correlated equilibrium $\dist^\star$ at contract $\alpha^\star$, and whether there is a ``significant'' agent, namely, an agent such that $\alpha^*_i > 3/4$ and $(1-\alpha_i^\star) \cdot \E[f(S_i^\star)] \geq 4 \cdot \E[f(S_{-i}^\star)]$, or not. If there is a significant agent, we show that we can get a good pure Nash equilibrium by 
% just 
incentivizing that agent alone. 

If there is no significant agent, 
then either (i) there is an agent with $\alpha^*_i > 3/4$ but $(1-\alpha_i^\star) \cdot \E[f(S_i^\star)] \leq 4 \cdot \E[f(S_{-i}^\star)]$, or (ii) $\alpha^*_i < 3/4$ for all agents $i$.
In case (i), we show that 
dropping agent $i$ and applying the Scaling-for-Existence Lemma to the remaining agents yields a good pure Nash equilibrium.
(Note that there can be at most one agent with $\alpha^\star_i > 3/4$ and that $\sum_{i' \neq i}\alpha^\star_{i'} \leq 1/4$.) 
In case (ii), we argue that the 
agents can be partitioned into two groups $B_1, B_2$ such that $\sum_{i' \in B_{\ell}} \alpha^\star_{i'} \leq 3/4$ for $\ell \in \{1,2\}$, and applying the Scaling-for-Existence Lemma to the better of the groups gives a good pure Nash equilibrium.

\begin{algorithm}[t]
\caption{Black-Box Lifting for XOS Rewards\label{alg:lift-xos}}
   \hspace*{\algorithmicindent} \textbf{Input:}  Costs $c_1,\ldots,c_m \in \reals_{\geq 0}$, value and demand oracle access to a XOS function $f:2^A \rightarrow \reals_{\geq 0}$, a contract $\contract^\star$, and a corresponding CCE $\dist^\star$. \\
    \hspace*{\algorithmicindent} \textbf{Output:}  A contract $\contract$, and a PNE $S$ with  $(1-\sum_i \alpha_i)f(S)\geq \Omega(1)(1-\sum_{i}\alpha_i^\star)E_{S^\star\sim\dist^\star}[f(S^\star)]$.
\begin{algorithmic}[1]
\If{  $\max_j \alpha^\star_{j} > 3/4 $}
\State Let $i =\arg\max_j \alpha^\star_{j}$ \Comment{There must be a unique maximum}
\If{$(1-\alpha_i^\star) \cdot \E_{S^\star\sim\dist^\star}[f(S_i^\star)] \geq 4 \cdot \E_{S^\star\sim \dist^\star}[f(S_{-i}^\star)]$}
\State Set $\contract$ such that $\alpha_{i} = \frac{1+\alpha_{i}^\star}{2} $ and $\alpha_j=0$ for $j \neq i$
\State Let $S \in \arg\max_{S'\subseteq A_{i}} (f(S')-\frac{c(S')}{\alpha_{i}})$
\State \Return $\contract,S$
\Else
\State Apply Algorithm~\ref{alg:lem}  
with $\dist^\star$, $\contract^\star$, $N' = [n]\setminus \{i\}$ and $\gamma=2$ to obtain $\contract$ and $S$
\State \Return $\contract,S$
\EndIf
\Else
\State Partition $N$ into two bundles $B_1,B_2$, where for each $\ell \in\{1,2\}$, $\sum_{i\in B_\ell} \alpha_i^\star \leq 3/4$
\State Let $\ell\in \arg\max_{\ell'\in\{1,2\} } \E_{S^\star\sim \dist^\star}[f(\bigcup_{i\in B_{\ell'}}S_i^\star)]$ 
\State Apply Algorithm~\ref{alg:lem}   
with $\dist^\star$, $\contract^\star$, $N' = B_\ell$ and $\gamma=\frac{7}{6}$ to obtain $\contract$ and $S$
\State \Return $\contract,S$
\EndIf
\end{algorithmic}
\end{algorithm}

\begin{proof}[Proof of Theorem~\ref{thm:constant-gap}]
Consider any contract $\contract^\star$ and any coarse-correlated equilibrium $\dist^\star$ with respect to $\contract^\star$, given as input to Algorithm~\ref{alg:lift-xos}. We analyze the guarantee provided by the contract computed by this algorithm. In the remainder of the proof, all expectations are over $S^\star$ that is distributed according to $\dist^\star$. 
We consider three cases:

\bigskip

\noindent \textbf{Case A:} There exists an agent $i$ with  $\alpha^\star_{i} > 3/4 $ and $(1-\alpha_i^\star) \cdot \E[f(S_i^\star)] \geq 4 \cdot \E[f(S_{-i}^\star)]$.
In this case, since $\alpha^\star_i > 3/4$, we have $\E[f(S_i^\star)] \geq 16 \cdot \E [f(S_{-i}^\star)]$. 
    By subadditivity of $f$, this implies that $\E[f(S_i^\star)] \geq 16 \cdot \E[(f(S^\star) - f(S_i^\star))]$, or equivalently, $\E[f(S_i^\star)] \geq \frac{16}{17} \cdot \E[f(S^\star)]$.  
    
    Consider contract $\contract$ with $\alpha_i= \frac{1 + \alpha^\star_{i}}{2}$ and $\alpha_{j} = 0$ for $j\neq i$. Note that $\alpha_{i} > \alpha^\star_i$. Also note that for agents $j \neq i$ doing nothing is a best response no matter what the other agents do. 
    Let $S = (S_i,\emptyset)$ be any pure Nash equilibrium of $\contract$ (such that $S'_{j} = \emptyset$ for all $j \neq i$). 
    We next show that $f(S_i) \geq \frac{1}{2} \cdot \E[f(S_i^\star)]$.
   
    First observe that since $S = (S_i, \emptyset)$ is a pure Nash equilibrium of $\contract$, we have
    \begin{align}
    \alpha_i f(S_i) - c(S_i) \geq 
    \E[\alpha_i f(S_i^\star) - c(S_i^\star)].
    \label{eq:one}
    \end{align}
    On the other hand, since $\mathcal{D}$ is a coarse-correlated equilibrium of $\contract^\star$, it must hold that 
    \[
    \E [\alpha_i^\star f(S_i^\star \mid S_{-i}^\star) - c(S_i^\star)] \geq \E [\alpha_i^\star f(S_i \mid S_{-i}^\star) - c(S_i)].
    \]
    By subadditivity of $f$, we have $f(S_i^\star \mid S_{-i}^\star) = f(S^\star_i \cup S_{-i}^\star) - f(S_{-i}^\star) \leq f(S^\star_i)$. By monotonicity of $f$, we have  $f(S_i \mid S^\star_{-i}) = f(S_i \cup S^\star_{-i}) - f(S^\star_{-i}) \geq f(S_i) - f(S^\star_{-i})$. We thus obtain
    \begin{align}
    \E[\alpha_i^\star f(S_i^\star) - c(S_i^\star)] \geq \E[\alpha_i^\star f(S_i) - \alpha_i^\star f(S_{-i}^\star) - c(S_i)].
    \label{eq:two}
    \end{align}
    By summing up~\eqref{eq:one} and~\eqref{eq:two}, using linearity of expectation, we get 
    \begin{equation}
    (\alpha_i - \alpha_i^\star) f(S_i) \geq (\alpha_i - \alpha_i^\star) \cdot \E[f(S_i^\star)] - \alpha_i^\star \cdot \E[f(S_{-i}^\star)]. \label{eq:alphaalpha}
    \end{equation}
    Thus, using that $\alpha_i > \alpha_i^\star$,
    \begin{eqnarray}
    f(S_i)  & \geq & \E[f(S_i^\star)] - \frac{\alpha_i^\star}{\alpha_i - \alpha_i^\star} \E[f(S_{-i}^\star)] = \E[f(S_i^\star)] - \frac{2\alpha_i^\star}{1 - \alpha_i^\star} \E[f(S_{-i}^\star)] \nonumber 
    \\ & \geq & \E[f(S_i^\star)] - \frac{2\alpha_i^\star}{1 - \alpha_i^\star}\cdot  \frac{1-\alpha_i^*}{4} \cdot \E[f(S_{i}^\star)] \geq \frac{1}{2} \E[f(S_{i}^\star)]  \label{eq:s'1-old},
    \end{eqnarray}
where the first inequality is by rearranging Inequality~\eqref{eq:alphaalpha}, the equality is by the definition of $\alpha_i$, the second inequality is by the assumption of the case, and the last inequality is since $\alpha_i^\star\leq 1$.

 Overall, the principal's utility under contract $\contract$ and equilibrium $S$ is
    \begin{align*}
    \bigg(1 - \sum_{j} \alpha_{j}\bigg) \cdot f(S) & = (1-\alpha_i) \cdot f(S_i) = \frac{1}{2}(1-\alpha_i^\star) \cdot f(S_i) \geq 
    \frac{1}{4}(1-\alpha_i^\star) \cdot \E[f(S_i^\star)] \\
    &\geq \frac{1}{4} \bigg(1 - \sum_{j} \alpha^\star_{j}\bigg) \cdot \E[f(S^\star_i)] \geq \frac{4}{17} \bigg(1 - \sum_{j} \alpha^\star_{j}\bigg) \cdot \E[f(S^\star)],
    \end{align*}
    where the second equality follows by the definition of $\alpha_i$. 
This concludes the argument for this case.

    \bigskip

    \noindent \textbf{Case B:} There exists an agent $i$ with $\alpha_i^\star > 3/4 $ and $(1-\alpha_i^\star) \cdot \E[f(S_i^\star)] \leq 4 \cdot \E[f(S_{-i}^\star)]$.
    In this case, let $\contract$ be the contract where $\alpha_i=0$ and $\alpha_j=2\alpha_j^\star$ for $j\neq i$.
    By applying Lemma~\ref{lem:double} on $\dist^\star$, $\contract^\star$, $N' = [n]\setminus \{i\}$ and $\gamma=2$ %and 
    we get that there exists a pure Nash equilibrium $S$ with respect to contract $\contract$ such that \begin{equation}
        f(S) \geq \frac{1}{2} \Ex{f(\bigcup_{j \in N'} S_j^\star)} . \label{eq:s'}
        \end{equation}
     We can bound the principal's utility under $\contract^\star$ and 
     $\dist^\star$ by 
     \begin{eqnarray}
     \left(1-\sum_j \alpha_j^\star\right)\Ex{f( S^\star)} & \leq & \left(1-\sum_j \alpha_j^\star\right)\Ex{f( S^\star_i)} + \left(1-\sum_j \alpha_j^\star\right)\Ex{f( S^\star_{-i})}   
     \nonumber \\ & \leq & \left(1- \alpha_i^\star\right)\Ex{f( S^\star_i)} + \Ex{f( S^\star_{-i})}   \leq 5 \cdot \Ex{f( S^\star_{-i})}, \label{eq:s-i}
     \end{eqnarray}
     where the first inequality is by subadditivity, and the last inequality is by the assumption of the case.

     One the other hand, under contract $\contract$ (for which $\sum_j \alpha_j = \sum_{j\neq i} 2\alpha^\star_j \leq  2(1- \alpha^\star_i) \leq \frac{1}{2}$), and equilibrium $S$, the principal's utility is  
     $$(1-\sum_j \alpha_j) f(S) \stackrel{\eqref{eq:s'}}{\geq }  (1- \sum_{j\neq i } 2\alpha^\star_j) \cdot   \frac{1}{2} \Ex{f(\bigcup_{j \in N'} S_j^\star)} \geq \frac{1}{4}   \Ex{f( S_{-i}^\star)} \stackrel{\eqref{eq:s-i}}{\geq} \frac{1}{20}   \left(1-\sum_j \alpha_j^\star\right) \Ex{f( S^\star)} ,$$
     which concludes the proof of the case.

    \bigskip

\noindent \textbf{Case C:} $\alpha_i^\star \leq 3/4$ for every $i$.
We claim that the agents can be partitioned into two bundles $B_1,B_2$, where for each $\ell \in\{1,2\}$, $\sum_{i\in B_\ell} \alpha_i^\star \leq 3/4$. To see this, consider the following process. 
We start by creating a separate bundle for each agent. Note that this way, by the assumption of the case, each bundle has sum of $\alpha^\star_i$ at most $3/4$. 
Then, as long as there are two bundles with sum of $\alpha^\star_i$ less than $3/4$ we merge them. This process is well-defined and terminates with two bundles with the desired property since as long as there are more than two bundles, since $\sum_i \alpha_i^\star \leq 1$, there must be two bundles with sum of $ \alpha_i^\star $ at most $2/3$, so we can merge two of them.

Now, assume without loss of generality that \begin{equation}
\E[f(\bigcup_{i\in B_1} S_i^\star)] \geq \E[f(\bigcup_{i\in B_2} S_i^\star)]. \label{eq:assum}
\end{equation}
Let $\contract$ be the contract where $\alpha_i=0$  for $i\in B_2$ and $\alpha_i=\frac{7}{6}\alpha_i^\star$ for $i\in B_1$.
    By applying Lemma~\ref{lem:double} on $\dist^\star$, $\contract^\star$, $N' = B_1$ and $\gamma=7/6$ we get that there exists an equilibrium $S$ with respect to contract $\contract$ such that \begin{equation}
        f(S) \geq \frac{1}{7} \Ex{f(\bigcup_{i \in B_1} S_i^\star)}.  \label{eq:s'2}
        \end{equation}

The principal's utility from $\contract$ and $S$ is 
\begin{eqnarray*}
     (1-\sum_i \alpha_i) f(S)  & \geq & (1-\frac{7}{6} \sum_{i\in B_1} \alpha_i^\star) \cdot \frac{1}{7}  \cdot \Ex{f(\bigcup_{i \in B_1} S_i^\star)}  \geq (1-\frac{7}{6} \cdot \frac{3}{4})  \cdot\frac{1}{7} \cdot  \Ex{f(\bigcup_{i \in B_1} S_i^\star)} \\ &  \geq  & \frac{1}{56} \frac{\Ex{f(\bigcup_{i \in B_1} S_i^\star)}+\Ex{f(\bigcup_{i \in B_2} S_i^\star)}}{2}  \geq \frac{\E[f(S^\star)]}{112},
     \end{eqnarray*}
where the first inequality is by the definition of $\contract$ and by Inequality~\eqref{eq:s'2}, the second inequality is since $\sum_{i\in B_1} \alpha^\star_j\leq 3/4$, the third inequality is since by  Inequality~\eqref{eq:assum}, and the last inequality is by subadditivity.
This concludes the proof of the theorem.
\end{proof}

\subsection{Robustness and Tractability}
\label{sec:robust-algo}

We next derive our robust approximation results. To this end, we show that for submodular rewards it is possible to turn any contract $\contract^\star$ and pure Nash equilibrium $S^\star$ of that contract into a contract $\contract$ such that any CCE under $\contract$ achieves a constant approximation to the principal's utility under $S^\star$.

\begin{theorem}[Black-Box Robustness Theorem]\label{thm:constant-gap-sub-CCE}
Let $f$ be a submodular reward function. 
There exists an algorithm (Algorithm~\ref{alg:robust-submod} in Appendix~\ref{app:constant-gap-sub-CCE}) that runs in polynomial time (in $n$ and $m$) using only value queries to $f$, that, given a contract
$\contract^\star$ and a corresponding PNE $S^\star$, outputs a contract
$\contract$ such that for every CCE $\dist$ of $\contract$,
\[
\left(1 - \sum_{i} \alpha_i \right)
\cdot \E_{S \sim \dist}[f(S)]
\;\ge\;
\Omega(1) \cdot
\left(1 - \sum_{i} \alpha_i^\star \right)
\cdot f(S^\star).
\]
\end{theorem}

Using \citep{DEFK23} and \citep{DEFK25} and combining with Theorem~\ref{thm:constant-gap} we obtain the following corollary:

\begin{corollary}[Efficient Robust Approximation Algorithms]
    \label{cor:robustness}
    For submodular rewards there is an algorithm that runs in polynomial time using value and demand queries and finds a contract $\contract$ such that any CCE of $\contract$ obtains an $O(1)$-approximation to the principal's optimal utility under the best CCE under any contract. For binary actions the same guarantee can be achieved with value queries only.
\end{corollary}

The key tool for proving Theorem~\ref{thm:constant-gap-sub-CCE} %this theorem 
is a Scaling-for-Robustness Lemma---a strengthened version of the Doubling Lemma from \cite{DEFK25}.\footnote{While the Doubling Lemma by \cite{DEFK25} applies to any PNE of the scaled contract, our Scaling-for-Robustness Lemma applies even with respect to any CCE of the scaled contract.}
This lemma plays a role orthogonal to the 
Scaling-for-Existence Lemma: rather than showing the existence of a contract and a PNE which is good relative to a reference CCE, 
it is used to derive a contract under
which \emph{every} CCE is good with respect to a reference PNE. Similar to the Scaling-for-Existence Lemma, rather than working directly
with the reference equilibrium, it starts from a dropout stable distribution over sets of actions.

\begin{lemma}
[Scaling-for-Robustness Lemma]
\label{lem:doubling-cce}\label{lem:scaling-for-robustness}
Suppose $f$ is submodular. Let $\epsilon > 0$ and let $\vec{\epsilon} = (\epsilon,\ldots,\epsilon) \in \reals_+^n$. Let $\gamma > 1$.  
Let $\dist$ be a dropout-stable distribution with respect to $\contract$.
Then any coarse correlated equilibrium (CCE) $\SSeq{}$ with respect to $\gamma \contract + \vec{\epsilon}$ satisfies $\E_{\Seq{} \sim \SSeq{}}[f(\Seq{})] \geq \frac{1}{2} (1 - \frac{1}{\gamma}) \E_{S \sim \dist}[f(S)]$. 
\end{lemma}

\begin{proof}
Let $\dist$ be a dropout-stable distribution with respect to $\contract$, 
and let $\SSeq{}$ be a CCE with respect to $\gamma\contract+\vec{\epsilon}$. 
{As $\SSeq{}$ is a CCE with respect to $\gamma \contract + \vec{\epsilon}$, agent $i$ weakly prefers taking action set $\Seq{}_i$ drawn from $\SSeq{}$ over $S_i$ drawn independently from $\dist$. 
That is,
\[
\E_{\Seq{} \sim \SSeq{}}
[(\gamma\alpha_i+\epsilon)f(\Seq{}) - c(\Seq{}_i)] \geq 
(\gamma\alpha_i+\epsilon)
\E_{\Seq{} \sim \SSeq{}, S \sim \dist}[f(\Seq{}_{-i} \cup S_i)] - c(S_i).
\]}
Rearranging the terms and taking the sum over all agents $i$, we obtain
\begin{equation}
\label{eq:exp-greater-0}
\E_{\Seq{} \sim \SSeq{}, S \sim \dist}
\left[\sum_{i = 1}^n f(\Seq{}_i\mid \Seq{}_{-i}) - \sum_{i = 1}^n f(S_i\mid  \Seq{}_{-i}) + \sum_{i = 1}^n \frac{c(S_i)- c(\Seq{}_i)}{\gamma\alpha_i+\epsilon}\right] \geq 0.
\end{equation}

By submodularity of $f$, for any $\Seq{}$ it holds that
\begin{equation}
\label{eq:eq1}
    \sum_{i = 1}^n f(\Seq{}_i\mid \Seq{}_{-i}) \leq f(\Seq{}),
\end{equation}
%and, \mfe{by submodularity again}, 
while for any $S,\Seq{}$ it holds that 
\begin{equation}
\label{eq:eq2}
\sum_{i = 1}^n f(S_i\mid  \Seq{}_{-i}) \geq f(S \setminus \Seq{} \mid \Seq{})
\end{equation}

Also, as $\dist$ is dropout-stable, we furthermore have for every agent $i$ that $\E_{S \sim \dist}[\alpha_i f(S) - c(S_i)] \geq \E_{S \sim \dist}[\alpha_i {f(S_{-i})}]$, implying that 
\begin{equation}
\label{eq:subset-stable}
    \frac{\E_{S \sim \dist}[c(S_i)]}{\alpha_i} \leq \E_{S \sim \dist}[f(S_i \mid S_{-i})].
\end{equation}
Therefore, it holds that
\begin{align}
\E_{\Seq{} \sim \SSeq{}, S \sim \dist}
\left[\sum_{i = 1}^n
  \frac{c(S_i)- c(\Seq{}_i)}{\gamma\alpha_i+\epsilon}\right] &\leq \sum_{i = 1}^n
  \frac{\E_{S \sim \dist}[c(S_i)]}{\gamma\alpha_i+\epsilon} \notag\\
  &\leq \frac{1}{\gamma} \sum_{i = 1}^n \frac{\E_{S \sim \dist}[c(S_i)]}{\alpha_i} 
  \leq \frac{1}{\gamma} \sum_{i = 1}^n \E_{S \sim \dist}[f(S_i \mid S_{-i})] \leq \frac{\E_{S \sim \dist}[f(S)]}{\gamma}, \label{eq:eq3}
\end{align}
where the next-to-last inequality follows by Equation~(\ref{eq:subset-stable}), and the final one holds by submodularity of $f$.

Combining Inequalities~(\ref{eq:eq1}), (\ref{eq:eq2}), we get that for any $S,\Seq{}$,
\[
\sum_{i = 1}^n f(\Seq{}_i\mid \Seq{}_{-i}) - \sum_{i = 1}^n f(S_i\mid  \Seq{}_{-i}) \leq f(\Seq{}) - f(S \setminus \Seq{} \mid \Seq{}) = 2 f(\Seq{}) - f(S \cup \Seq{}) \leq 2 f(\Seq{}) - f(S).
\]

Taking expectation over the last inequality and combining it with Equation~(\ref{eq:exp-greater-0}) and (\ref{eq:eq3}) gives 
\begin{align*}
0 &\leq
\E_{\Seq{} \sim \SSeq{}, S \sim \dist}
\left[\sum_{i = 1}^n f(\Seq{}_i\mid \Seq{}_{-i}) - \sum_{i = 1}^n f(S_i\mid  \Seq{}_{-i}) + \sum_{i = 1}^n \frac{c(S_i)- c(\Seq{}_i)}{\gamma\alpha_i+\epsilon}\right] \\
&
\leq \E_{\Seq{} \sim \SSeq{}, S \sim \dist}\left[2 f(\Seq{}) - (1 - \frac{1}{\gamma}) f(S)\right],
\end{align*}
yielding $\E_{\Seq{} \sim \SSeq{}, S \sim \dist}\left[2 f(\Seq{}) - (1 - \frac{1}{\gamma}) f(S)\right] \geq 0$, and thus $\E_{\Seq{} \sim \SSeq{}}[f(\Seq{})] \geq \frac{1}{2} (1 - \frac{1}{\gamma}) \E_{S \sim \SSeq{}}[f(S)]$, as desired.
\end{proof}

The proof of Theorem~\ref{thm:constant-gap-sub-CCE} concludes by applying Lemma~\ref{lem:doubling-cce} to analyze the contract returned by Algorithm~\ref{alg:robust-submod} (see Appendix~\ref{app:constant-gap-sub-CCE} for details).

%% SUBADDITIVE REWARDS
\section{Subadditive Rewards}

In this section, we show that for subadditive rewards, there is a $\Theta(\text{poly(n)})$ gap between the principal's utility under the best coarse-correlated equilibrium and the best pure Nash equilibrium. 
We first show that this gap arises already between mixed and pure Nash equilibria, and even with binary actions.

\begin{proposition}[Subadditive Rewards]\label{prop:subadditive-binary}
There exists a binary-action instance with a subadditive reward function
in which the gap between the principal's utility 
under the best MNE and under the best PNE is $\Omega(\sqrt{n})$.
\end{proposition}
\begin{proof}
Consider a setting with $2n+2$ agents denoted by $A=\{x,y\} \cup [2n]$.
We define the reward function $f: 2^A \rightarrow \reals_{\ge 0}$ in Table~\ref{tab:subadditive-function}. The costs of agents $x$ and $y$  are $0$, and the costs of all remaining agents are $c_i = \frac{2}{3n}$. 
\begin{table}[ht]
\centering
\renewcommand{\arraystretch}{1.5} 
  \begin{tabular}{|c|c|c|c|c|c|}
    \hline
    \multirow{2}{*}{\diagbox{\centering $|S\cap \{x,y\}|$}{\centering $|S\cap [2n]|$}} 
    & \multirow{2}{*}{0} 
    & \multirow{2}{*}{$0 < i < 2n-1$} 
    & \multicolumn{2}{|c|}{$2n-1$} 
    & \multirow{2}{*}{$2n$} \\ \cline{4-5}
    & & & $[n]\subseteq S$ & $[n]\not\subseteq S$ & \\ \hline
    0 & $0$ & $2+\frac{i}{\sqrt{n}}$ & $2+\frac{2n-1}{\sqrt{n}}$ & $3+\frac{2n-1}{\sqrt{n}}$ & $4+\frac{2n-1}{\sqrt{n}}$ \\ \hline
    1 & $4$ & $4+\frac{i}{\sqrt{n}}$ & $5+\frac{2n-1}{\sqrt{n}}$ & $4+\frac{2n-1}{\sqrt{n}}$ & $6+\frac{2n-1}{\sqrt{n}}$ \\ \hline
    2 & $5$ & $5+\frac{i}{\sqrt{n}}$ & $5+\frac{2n-1}{\sqrt{n}}$ & $6+\frac{2n-1}{\sqrt{n}}$ & $7+\frac{2n-1}{\sqrt{n}}$ \\ \hline
  \end{tabular}
\caption{Definition of the monotone subadditive function $f:2^{[2n]\cup\{x,y\}} \rightarrow \mathbb{R}_{\geq 0}$. The value of a set $S$ is determined by $|S\cap\{x,y\}|$ and $|S\cap [2n]|$, and in the case $|S\cap[2n]|=2n-1$, by whether the missing element belongs to $[n]$.}
\label{tab:subadditive-function}
\end{table}

\begin{claim}
The function    $f$ is monotone and subadditive.\label{cl:sub-mon}
\end{claim}
\begin{proof}
We first prove that the function $f$ is monotone. Afterwards, we show that it is subadditive.

\paragraph{Monotonicity.} To see that $f$ is monotone, we can observe that all of the marginals (presented in Tables~\ref{tab:margx} and \ref{tab:margi}) are non-negative. 

\begin{table}[ht]
\centering
\renewcommand{\arraystretch}{1.5} 
  \begin{tabular}{|c|c|c|c|c|c|}
    \hline
    \multirow{2}{*}{\diagbox{\centering $|S\cap \{x,y\}|$}{\centering $|S\cap [2n]|$}} 
    & \multirow{2}{*}{0} 
    & \multirow{2}{*}{$0<i<2n-1$} 
    & \multicolumn{2}{|c|}{$2n-1$} 
    & \multirow{2}{*}{$2n$} 
    \\ \cline{4-5}
    & & & $[n] \subseteq S$ & $[n] \not\subseteq S$ & \\ \hline
    0 & $4$ & $2$ & $3$ & $1$& $2$\\ \hline
    1 & $1$   &       $1$             &$0$ & $2$& $1$ \\ \hline
  \end{tabular}%
\caption{The marginal $f(j\mid S)$ for $j\in \{x,y\} $.\label{tab:margx}}
\end{table}

\begin{table}[ht]
\centering
\renewcommand{\arraystretch}{1.5} 
  \begin{tabular}{|c|c|c|c|c|c|c|}
    \hline
    \multirow{2}{*}{\diagbox{\centering $|S\cap \{x,y\}|$}{\centering $|S\cap [2n]|$}} 
    & \multirow{2}{*}{0} 
    & \multirow{2}{*}{$0<i < 2n-2$} 
    & \multicolumn{2}{|c|}{$2n-2$} 
    & \multicolumn{2}{|c|}{$2n-1$} 
    \\ \cline{4-7}
    & & & $[n] \setminus \{j\} \subseteq S$ & $[n]\setminus\{j\} \not\subseteq S$ & $j\in [n] $ & $ j\notin [n]$ \\ \hline
    0 & $2+\frac{1}{\sqrt{n}}$ & $\frac{1}{\sqrt{n}}$ & $\frac{1}{\sqrt{n}}$ & $1+\frac{1}{\sqrt{n}}$& $1 $ & 2\\ \hline
    1 & $\frac{1}{\sqrt{n}}$   &       $\frac{1}{\sqrt{n}}$             &$1+\frac{1}{\sqrt{n}}$ & $\frac{1}{\sqrt{n}}$& $2 $ & 1 \\ \hline
    2 & $\frac{1}{\sqrt{n}}$   &   $\frac{1}{\sqrt{n}}$                 & $\frac{1}{\sqrt{n}}$& $1+\frac{1}{\sqrt{n}}$ & $1$  & 2 \\ \hline
  \end{tabular}%
\caption{The marginal $f(j\mid S)$ for $j\in [2n]$.  \label{tab:margi}}
\end{table}

\paragraph{Subadditivity.} To see that $f$ is subadditive, observe that for every $T \neq \emptyset$, it holds that \begin{equation}
        \label{eq:lowerf} f(T)\geq 2 +2\cdot\indicator{|\{x,y\} \cap T | \geq 1 } +\indicator{|\{x,y\} \cap T | \geq 2 }+ \frac{|T\setminus\{x,y\}|}{\sqrt{n}}, \end{equation}
    on the other hand \begin{equation}
        \label{eq:upperf}f(T)\leq 4 +2\cdot\indicator{|\{x,y\} \cap T|\geq 1 } +\indicator{|\{x,y\} \cap T | \geq 2 }+ \frac{|T\setminus\{x,y\}|}{\sqrt{n}}. \end{equation}
    Thus, \begin{eqnarray*}
 f(S) + f(T) &  \geq & 2+2+2\cdot\indicator{|\{x,y\} \cap S|\geq 1 } +2\cdot\indicator{|\{x,y\} \cap T | \geq 1 } \\ & + & \indicator{|\{x,y\} \cap S|\geq 2 } +\indicator{|\{x,y\} \cap T | \geq 2 }+\frac{|S\setminus\{x,y\}|}{\sqrt{n}} +\frac{|T\setminus\{x,y\}|}{\sqrt{n}}
 \\ & \geq & 4+ 2\cdot\indicator{|\{x,y\} \cap (S\cup T)|\geq 1 } + \indicator{|\{x,y\} \cap (S\cup T)|\geq 2 } +  \frac{|S\cup T\setminus\{x,y\}|}{\sqrt{n}} \\
 \\ & \geq & f(S\cup T),
     \end{eqnarray*}
   where the first inequality is by Inequality~\eqref{eq:lowerf}, the second inequality is by subadditivity of the indicator function, and since $\frac{|S\cap [2n]|}{\sqrt{n}}$ is an additive function, and the last inequality is by Inequality~\eqref{eq:upperf}  
     This concludes the proof of the claim.
\end{proof}

\begin{claim}
    No PNE achieves a principal utility of more than $6.5$.\label{cl:sub-pne}
\end{claim}
\begin{proof}
    Let $S$ be some PNE. If $ |S\cap [2n]| \leq 1$ then the utility of the principal is bounded by the reward which is bounded by $5 +\frac{1}{\sqrt{n}}$. If $1<|S\cap [2n]| < 2n-1$ then the marginal of all agents in $S\cap [2n]$ are $\frac{1}{\sqrt{n}}$, which means that the principal's utility is bounded by $f(S) (1-|S\cap [2n]|\cdot \frac{c_i}{\sqrt{n}})$, which is negative for $|S\cap[2n]| >\frac{3\sqrt{n}}{2}$, thus, the utility of principal is bounded by $5+\frac{\nicefrac{3\sqrt{n}}{2}}{\sqrt{n}} =6.5$. For $S$ with $|S\cap[2n]|=2n-1$, the marginals of all the $2n-1$ agents in $|S\cap[2n]$ is bounded by $1+\frac{1}{\sqrt{n}}$, thus the fraction that remains with the principal is $1-(2n-1)\frac{\nicefrac{2}{3n}}{1+\frac{1}{\sqrt{n}}} <0 , $ which means that the principal obtains a negative utility from incentivizing this set. If $|S\cap[2n]|=2n$, then half of the agents (in $[2n]$ have a marginal of 1, and the remaining half have marginals of $2$. Thus, the fraction of the reward that remains with the principal is $ 1- n\frac{\nicefrac{2}{3n}}{1} -n\frac{\nicefrac{2}{3n}}{2} =0 $, which means that the principal cannot obtain positive utility from incentivizing this set.
\end{proof}

\begin{claim}
    There exists a MNE that achieves a principal utility of  $\Omega(\sqrt{n})$.\label{cl:sub-mne}
\end{claim}
\begin{proof}
    Consider the MNE where all agents in $[2n]$ take action with a probability of $1$, and agents $x,y$ take action (each) with a probability of $\frac{1}{2}$.
    Since the costs of $x,y$ are zero, we could use $\alpha_x=\alpha_y=0$. For the remaining agents, the expected marginal contribution of them to the reward is $1.5$ since with probability $\frac{1}{2}$ their marginal contribution is 1, and with probability half their marginal contribution is $2$. Thus, it is sufficient to use contract $\frac{c_i}{1.5} = \frac{4}{9n}$ to incentivize them. Overall, the utility of the principal from this MNE is $E[f(S)](1-2n\cdot\frac{4}{9n}) = (\frac{1}{4}\cdot 4+ \frac{1}{2}\cdot 6 + \frac{1}{4}\cdot 7 + \frac{2n-1}{\sqrt{n}})\cdot \frac{1}{9} = \Omega(\sqrt{n}) $, which concludes the proof.
\end{proof}
The proof of the proposition follows from Claims~\ref{cl:sub-mon}, \ref{cl:sub-pne}, and \ref{cl:sub-mne}.
\end{proof}

We next show that the gap between the principal's utility under the best CCE and the best PNE is at most polynomial in the number of agents.

\begin{proposition}
    \label{prop:subadditive-upper-bound}
The gap between the principal's utility
under the best CCE and under the best PNE is at most $O(n)$. 
\end{proposition}
The proof of Proposition~\ref{prop:subadditive-upper-bound} is deferred to Appendix~\ref{app:proof-subadditive}.

    \begin{remark}
    We note that \citet{DEFK25} show that the %gap between the 
    principal's utility under the best PNE is at least a $\frac{1}{m}$-fraction of the optimal welfare (which trivially bounds the principal's utility under the best CCE). This implies an upper bound of $m$ on the gap between the utility under the best CCE and the best PNE for subadditive reward functions. Proposition~\ref{prop:subadditive-upper-bound} establishes a stronger upper bound that does not depend on the number of actions, but only on the number of agents.
    \end{remark}

%% SUPERMODULAR REWARDS
\section{Supermodular Rewards}

In this section, we establish our results for supermodular rewards. We first show that with binary actions, there is no gap between the principal's utility in the best coarse-correlated equilibrium and the best pure Nash equilibrium. 

\begin{theorem}[Supermodular Rewards, Binary Actions]\label{thm:no-gap-super-binary} 
There exists no gap between the principal's utility 
under the best CCE and the best PNE in binary-action settings.
\end{theorem}
\begin{proof} 
Consider a contract $\contract$, and a CCE $\dist$. Let $S'$ be the union of the sets of agents that are in the support of $\dist$. 
Let $\contract'$ be the contract for which $\alpha_i' = 0$ for $i \not\in S'$, and $\alpha_i'=\alpha_i$ for $i\in S'$.
We prove that $S'$ is a pure Nash equilibrium with respect to $\contract'$.

In the remainder of the proof, all expectations (and probabilities) are over $S$, which is distributed according to $ \dist$.
As $\dist$ is a coarse-correlated equilibrium, no agent can improve their utility by unilaterally not working anymore. That is, for all $i \in S'$, we have
    $$\E[\alpha_i \cdot f(S) - c_i \indicator{i \in S}] \geq \E[\alpha_i \cdot f(S \setminus \{i\})],$$ or equivalently $$\E[\alpha_i \cdot (f(S) - f(S \setminus \{i\})) - c_i \indicator{i \in S}] \geq 0.$$ For every set $S$ in the support of $\dist$ and every agent $i\in A$ we have
    \begin{align*}
         f(i\mid S \setminus \{i\}) = (f(S) - f(S \setminus \{i\}))\cdot \indicator{i \in S} \leq  f(i \mid S' \setminus \{i\}) \cdot\indicator{i \in S},
    \end{align*}
    where the last inequality is since $f$ is supermodular, and $S\subseteq S'$.
    In combination, this means that
    \[
    \Pr{i \in S} ( \alpha_i \cdot f(i\mid S' \setminus \{i\}) - c_i ) = \E[\alpha_i \cdot  f(i \mid S' \setminus \{i\}) \cdot \indicator{i \in S} - c_i \cdot \indicator{i \in S}] \geq 0.
    \]
    So, if $i \in S'$, then $\alpha_i \cdot f(i \mid S' \setminus \{i\}) - c_i \geq 0$, meaning that $\alpha_i f(S') - c_i \geq \alpha_i f(S' \setminus \{i\})$, and therefore $\alpha_i' f(S') - c_i \geq \alpha_i' f(S' \setminus \{i\})$, which means that agent $i$ does not want to deviate. For $i\not\in S'$, it holds that $\alpha_i' \cdot f(i\mid S') =0\leq c_i$, which means that agent $i$ does not want to deviate. Overall, we deduce that $S'$ is a pure Nash equilibrium with respect to contract $\contract'$.
\end{proof}

We next consider the general (non-binary) setting for which we show that there is no gap between CE and PNE.

\begin{theorem}[Supermodular Rewards, General Actions]\label{thm:no-gap-super-multi} 
There exists no gap between the utility of the principal under the best CE and the best PNE.
\end{theorem}

\begin{proof}
Consider a contract $\contract$ and a correlated equilibrium $\dist$ with respect to this contract.
Let $T=\cup_{S\in \sup(\dist)} S$. 
Before proving the theorem we are first going to prove the following key lemma.
\begin{lemma}\label{lem:local}
    For every agent $i$, assuming that every agent $j\neq i$ selects a set of actions $S_j$ such that $T_j \subseteq S_j $ (where we denote by $S_{-i} =\cup_{j \neq i} S_j$), then there exists a set $S_i$ for which  $T_i \subseteq S_i $ that is agent $i$'s best response. I.e., $$ S_i \in \arg\max_{X_i \subseteq A_i} \alpha_i f(X_i \cup S_{-i}) -c(X_i).$$
\end{lemma}
\begin{proof}
    Let $S_i^1,\ldots,S_i^k$ be the support of the sets taken by agent $i$ according to $\dist$ (ordered arbitrarily).
    Consider an arbitrary set  $R_i$ in $ \arg\max_{X_i \subseteq A_i} \alpha_i f(X_i \cup S_{-i}) -c(X_i)$.
    For $j=0,\ldots,k$, let $R_i^j = R_i \cup \bigcup_{\ell =1}^j S_i^\ell $ (where $R_i^0=R_i$).
    We denote by $\dist_{-i}(S_i^j)$ the distribution of the set of actions suggested to agents $N \setminus \{i\}$, conditioned on agent $i$ suggested the action set $S_i^j$.
    By that $\dist$ is a correlated equilibrium, we have that $$\alpha_i\cdot \E_{X_{-i} \sim \dist_{-i} (S_i^j)} [f(S_i^j \cup X_{-i})] -c(S_{i}^j) \geq \alpha_i \cdot  \E_{X_{-i} \sim \dist_{-i} (S_i^j)} [f((S_i^j \cap  R_i^{j-1}) \cup X_{-i})] -c(S_{i}^j\cap R_i^{j-1}) , $$
    or by rearranging, and since $S_i^j\setminus R_i^{j-1} =R_i^j\setminus R_i^{j-1}$ we get that
    $$ \alpha_i \cdot \E_{X_{-i} \sim \dist_{-i} (S_i^j)} [f(R_i^j\setminus R_i^{j-1} \mid (S_i^j \cap R_i^{j-1})\cup X_{-i})] \geq c(R_{i}^j \setminus R_i^{j-1}). $$
    Now, since (1) $X_{-i} \subseteq S_{-i}$ for each realization in the support of $\dist_{-i}(S_i^j)$, (2) $S_i^j \cap R_i^{j-1} \subseteq R_i^{j-1}$,  (3) $ R_i^j\setminus R_i^{j-1}$ is disjoint from $R_{i}^{j-1}\cup S_{-i}$, and (4) $f$ is supermodular, we deduce that 
    $$ \alpha_i \cdot \E_{X_{-i} \sim \dist_{-i} (S_i^j)} [f(R_i^j\setminus R_i^{j-1} \mid  R_i^{j-1}\cup S_{-i})] \geq c(R_{i}^j \setminus R_i^{j-1}).$$
    Since the last expression is deterministic (as it does not depend on the realization of $X_{-i}$), we conclude that 
    $$\alpha_i \cdot f(R_i^j\setminus R_i^{j-1} \mid  R_i^{j-1}\cup S_{-i}) \geq c(R_{i}^j \setminus R_i^{j-1}).$$
    By summing over $j\in [k]$ we obtain that 
        $$\alpha_i \cdot ( f(R_i^k \cup S_{-i}) - f(R_i^0 \cup S_{-i}) )   \geq c(R_{i}^k \setminus R_i^{0}),$$
        which by rearrangement we conclude that 
        $$\alpha_i \cdot  f(R_i^k \cup S_{-i}) - c(R_i^k) \geq \alpha_i \cdot f(R_i^0 \cup S_{-i})   - c( R_i^{0}).$$
        Since $R_i^0 = R_i$, and $R_i \in \arg\max_{X_i \subseteq A_i} \alpha_i f(X_i \cup S_{-i}) -c(X_i)$ 
        this implies that $$R_i^k  \in \arg\max_{X_i \subseteq A_i} \alpha_i f(X_i \cup S_{-i}) -c(X_i).$$
        This concludes the proof of the lemma since $T_i = \cup_j S_i^j \subseteq  R_i^k$.       
   \end{proof}
We are now ready to prove the theorem.
For this, we observe that if we start with $T$, and as long as some agent $i$ can strictly improve her utility, she improves her utility to a set containing $T_i$ (which is without loss because of Lemma~\ref{lem:local}), this process maintains the property that the action profile always contains $T$. Moreover, the process must terminate after a finite number of steps, as at each step, the value of the potential function (see Equation~\eqref{eq:potential}) 
 strictly increases, and this is a finite function.  
\end{proof}

We next prove that Theorem~\ref{thm:no-gap-super-multi} is tight by presenting an instance with a supermodular reward function,  for which there is an unbounded gap in the principal's utility that can be obtained from a CCE compared to a PNE.

\begin{proposition}[Supermodular Rewards, General Actions]
\label{prop:supermodular-gap-for-cce}
There exists an instance in which the gap between the principal's utilities under the best CCE and under the best PNE is unbounded.
\end{proposition}
\begin{proof}
Consider an instance with two agents where 
$A_1 = \{1,2\}$, and $A_2=\{3\}$.
Let $f:2^A\rightarrow \reals_{\geq 0}$ be as defined 
as in Table~\ref{tab:f}. The function $f$ is monotone and supermodular.
The costs of the actions are $c(1) = 0.75$, $    c(2) = 8.5$ and $c(3) =0.25$.

\begin{table}[h!]
\centering
\begin{tabular}{|c|c|c|c|c|c|c|c|c|}
\hline
$S$ & $\emptyset$ & $\{1\}$ & $\{2\}$ & $\{3\}$ & $\{1,2\}$ & $\{1,3\}$ & $\{2,3\}$ & $\{1,2,3\}$ \\ \hline
$f(S)$   & 0   & 0   & 1   & 0   & 5.5   & 1   & 1   & 10   \\ \hline
\end{tabular}
\caption{A supermodular reward function $f$ with unbounded gap between the best CCE and the best PNE.\label{tab:f}}
\end{table}

Note that the only set of actions with a positive welfare is $\{1,2,3\}$. Thus, this is the only candidate for a PNE with a positive utility for the principal. To incentivize $\{1,2,3\}$, the principal must use a contract $\contract$ such that $\alpha_1\cdot f(\{1,2\} \mid \{3\}) - c(\{1,2\}) \geq \alpha_1 \cdot f(1\mid \{3\}) -c(1)$ and $\alpha_2 f(3\mid \{1,2\}) \geq c(3)$, thus $\alpha_1\geq 8.5/9 = 17/18 $, and $\alpha_2 \geq 0.25/4.5 = 1/18 $, thus the principal cannot obtain a positive utility.

Consider the contract $\contract_3=(0.925,1/18)$. It holds that the distribution $\dist$ where $\{1,2,3\}$ is suggested with probability $0.8$, and $\emptyset$ is suggested with probability $0.2$ is a CCE.
Indeed this is a CCE since the utility of agent 2 when following the suggestion is $0.8\cdot (\alpha_2\cdot f(\{1,2,3\})-c(3))= 11/45 $, while if agent 2 does nothing, his utility is $0.8(\alpha_2 \cdot f(\{1,2\})) =11/45$, and if he always takes action 3, his utility is $ 0.8\alpha_2 f(\{1,2,3\}) +0.2\alpha_2 f(3)-c(3) =7/36 < 11/45$. For agent 1, following the suggestion of the principal leads to a utility of $0.8 \cdot ( 0.925 \cdot 10 - 9.25) = 0 $.
Now deviating to $\emptyset$ leads to a utility of $0$, deviating to $\{1\}$ leads to a utility of $0.8\cdot 0.925\cdot  1 - 0.75 = -0.01 $, deviating to $\{2\}$ leads to a utility of $ \alpha_1 \cdot 1 -c(2) = -7.575$, and deviating to $\{1,2\}$ leads to a utility of $0.8\cdot 0.925\cdot 10+0.2\cdot 0.925\cdot 5.5 - 9.25 = -0.8325 $. Thus, it is a CCE.

The utility of the principal under this CCE is $(1-0.925-1/18)\cdot 0.8 \cdot 10 = 7/45>0$, which concludes the proof.
\end{proof}

\section{General Rewards}

In this section, we show a lower bound on the gap between the principal's utility from the best MNE and the best PNE for general reward functions (that are neither subadditive, nor supermodular). We show that this gap is unbounded, even for a constant number of agents and binary actions.

\begin{proposition}[General Rewards]\label{prop:general-binary}
There exists a binary-action instance in which the gap between the principal's utility 
under the best MNE and under the best PNE is unbounded.
\end{proposition}

\begin{proof}
Consider a binary-action instance with four agents, i.e., $A=\{1,2,3,4\}$. 
The reward function is defined by the monotone closure of the following values (i.e., the maximum over all defined subsets). Let
\[
f(\{1,2\}) = 2 , f(\{1,3\}) = f(\{2,4\}) = 1, f(\{1,2,3\}) = f(\{1,2,4\}) = \phi+1,
\]
where $\phi = \frac{1+\sqrt{5}}{2} \approx 1.618$ (i.e., $\phi$ is the golden ratio).
Note that $f$ is  not subadditive (as $f(1)+f(2) <f(\{1,2\}$), nor supermodular (as $f(\{1,2,3\} +f(\{1,2,4\}) \geq f(\{1,2\} +f(\{1,2,3,4\})$).

The costs are
\[
c(1) = c(2) =1, \text{ and } c(3) = c(4)=0.
\]
We next calculate the utility of the principal under a PNE corresponding to all non-redundant sets of agents (sets that do not contain an agent with a $0$ marginal, and non-zero cost).
For this, we utilize the characterization of \cite{BFN06} for the principal's utility function as a function of the PNE for binary instances: $$ g(S)=f(S)(1-\sum_{i \in S} \frac{c(i)}{f(i\mid S\setminus \{i\})}).$$
It holds that: 
$$ g(\{1,2\}) = f(\{1,2\}) (1-\frac{c(1)}{f(1\mid \{2\})} -\frac{c(2)}{f(2\mid \{1\})} )  = 0 $$
$$ g(\{1,3\}) = f(\{1,3\}) (1-\frac{c(1)}{f(1\mid \{3\})}  )  = 0 $$
$$ g(\{1,3,4\}) = f(\{1,3,4\}) (1-\frac{c(1)}{f(1\mid \{3,4\})}  )  = 0 $$
$$ g(\{2,4\}) = f(\{2,4\}) (1-\frac{c(2)}{f(2\mid \{4\})}  )  = 0 $$
$$ g(\{2,3,4\}) = f(\{2,3,4\}) (1-\frac{c(2)}{f(2\mid \{3,4\})}  )  = 0 $$

$$ g(\{1,2,3\}) = f(\{1,2,3\}) (1-\frac{c(1)}{f(1\mid \{2,3\})} -\frac{c(2)}{f(2\mid \{1,3\})}  )  =f(\{1,2,3\}) (1-\frac{1}{\phi+1} -\frac{1}{\phi}  )  = 0 $$

$$ g(\{1,2,4\}) = f(\{1,2,4\}) (1-\frac{c(1)}{f(1\mid \{2,4\})} -\frac{c(2)}{f(2\mid \{1,4\})}  )  =f(\{1,2,4\}) (1-\frac{1}{\phi} -\frac{1}{\phi+1}  )  = 0 $$
$$ g(\{1,2,3,4\}) = f(\{1,2,3,4\}) (1-\frac{c(1)}{f(1\mid \{2,3,4\})} -\frac{c(2)}{f(2\mid \{1,3,4\})}  )  =f(\{1,2,3,4\}) (1-\frac{1}{\phi} -\frac{1}{\phi}  )  <  0 $$

Thus, no PNE obtains a positive utility for the principal.

Consider the contract $\contract = (\frac{4}{5\phi},\frac{4}{5\phi},0,0)$, and the distribution $\dist$ over agents where agents 1,2 always take an action with probability 1, and agents 3 and 4 take an action with probability $1-\frac{\phi}{2}$ each (independently).
The expected marginal of agent 1 (similarly, of agent 2) is:
$(1-\frac{\phi}{2})^2 \cdot \phi + (\frac{\phi}{2})(1-\frac{\phi}{2}) \cdot \phi + (1-\frac{\phi}{2})(\frac{\phi}{2}) (\phi+1) +(\frac{\phi}{2})^2 \cdot 2 = \frac{5\phi}{4}$.
Distribution $\dist$ is a MNE since agent 1 (similarly, agent 2) is paid $\frac{c(1)}{E_{S\sim \dist}[f(1\mid S\setminus \{1\})]} = \frac{4}{5\phi}$.
The principal's utility under this MNE is $$  (1-2\cdot\frac{4}{5\phi}) (\frac{\phi+1}{4}\cdot 2 +(1-\frac{\phi+1}{4})\cdot (\phi+1)) >0 , $$
which concludes the proof.
\end{proof}

%% BIBLIOGRAPHY
\printbibliography

\appendix

%% ADDITIONAL SEPARATION EXAMPLES
\section{Separation Example}

The following example (adopted from \cite[][Example 3.1]{BabaioffFN10}) with two identical agents each having a single action, shows that with submodular rewards the principal can strictly benefit from inducing a mixed Nash equilibrium rather than a pure Nash equilibrium. 
%\mfc{Moved here:}
We remark that the reward function in this example is also gross-substitutes and not just submodular as every submodular function over two elements is also gross-substitutes.

\begin{example}[Separation Example]
\label{ex:mixed}
Consider an instance with two identical agents $A=\{1,2\}$ with costs $c(1) = c(2) = c = 1.$ The submodular reward function is such that  $f(\emptyset) = 0$, $f(1) =f(2) = 180$, and $f(\{1,2\}) = 200$.

Recall the definition of the function $g: 2^A \rightarrow \reals$, introduced in the proof of Proposition~\ref{prop:general-binary}. 
Under pure Nash equilibria, the best contracts for the possible action profiles $\emptyset$, $\{1\}$, $\{2\}$, and $\{1,2\}$  yield the principal a utility of
\begin{align*}
&g(\emptyset) = 0, \\
&g(1) = f(1) \cdot \left(1- \frac{c}{f(1)} \right) = 179, \quad\\
&g(2) = f(2) \cdot \left(1- \frac{c}{f(2)} \right) = 179, \quad\text{and}\\
&g(\{1,2\}) = f(\{1,2\}) \cdot \left(1 -\frac{c}{f(1 \mid \{2\})}- \frac{c}{f(2 \mid \{1\})}\right) = 180.
\end{align*}
Thus, the highest utility the principal can obtain in a pure Nash equilibrium is $180$, by inducing both agents to exert effort.

Consider the contract $\contract = (\frac{1}{36},\frac{1}{36})$. We show that the product distribution $\dist$ in which each agent takes an action with probability $0.9$ is a MNE.  Assuming that one agent takes an action with probability $0.9$, then the utility of the other action as a function of the probability it takes an action is denoted by $$ u_a(p) = \frac{1}{36} \left(p  \cdot (0.9 \cdot f(\{1,2\}) + 0.1 \cdot f(1)) + (1-p) \cdot (0.9 \cdot f(1) + 0.1 \cdot f(\emptyset)) \right) -p \cdot c  = 4.5,  
$$
which means that the utility of the agent is independent of the probability of the agent taking action, which means that $\dist $ is indeed an MNE.

The expected principal's utility under contract $\contract$  and distribution $\dist$ is
\[
\E_{S\sim \dist}[u_p(\contract,S)] = (1-\frac{2}{36}) \left(\frac{1}{100} \cdot f(\emptyset) +\frac{9}{100}\cdot f(1) + \frac{9}{100}\cdot f(2) +\frac{81}{100}\cdot  f(\{1,2\}) \right)= \frac{17}{18} \cdot 194.4 = 183.6,
\]
which is strictly higher than the optimal principal utility in a pure Nash equilibrium. 
This shows that the principal can strictly gain from inducing a mixed Nash equilibrium, rather than a pure Nash equilibrium.
\end{example}

%% Proof of Theorem 3.3
\section{Proof of Theorem~\ref{thm:constant-gap-sub-CCE}}
\label{app:constant-gap-sub-CCE}

The high-level idea is to 
distinguish cases
based on the equilibrium $S^\star$ at contract $\contract^\star$, and whether it is sufficient to incentivize only zero cost actions, or
whether there is a ``significant'' agent, namely, an agent such that $\alpha^*_i > 3/4$ and $(1-\alpha_i^\star) \cdot f(S_i^\star) \geq 4 \cdot f(S_{-i}^\star)$, or not. If there is a significant agent, we show that we can get a good approximation under any CCE by aiming to 
incentivize that agent alone. 
If there is no significant agent, 
then either (i) there is an agent with $\alpha^*_i > 3/4$ but $(1-\alpha_i^\star) \cdot f(S_i^\star) \leq 4 \cdot f(S_{-i}^\star)$, or (ii) $\alpha^*_i < 3/4$ for all agents $i$. 
In case (i), we show that 
dropping agent $i$ and applying the Scaling-for-Robustness Lemma to the remaining agents yields a good guarantee under any CCE.
(Note that there can be at most one agent with $\alpha^\star_i > 3/4$ and that $\sum_{i' \neq i}\alpha^\star_{i'} \leq 1/4$.) 
In case (ii), we argue that the 
agents can be partitioned into two groups $B_1, B_2$ such that $\sum_{i' \in B_{\ell}} \alpha^\star_{i'} \leq 3/4$ for $\ell \in \{1,2\}$, and applying the Scaling-for-Robustness Lemma to the better of the groups guarantees a good principal's utility under any CCE.

\begin{algorithm}[t]
\caption{Black-Box Robustness for Submodular Rewards\label{alg:robust-submod}}
   \hspace*{\algorithmicindent} \textbf{Input:}  Costs $c_1,\ldots,c_m \in \reals_{\geq 0}$, value oracle access to a submodular function $f:2^A \rightarrow \reals_{\geq 0}$, a contract $\contract^\star$, and a corresponding PNE $S^\star$. \\
    \hspace*{\algorithmicindent} \textbf{Output:}  A contract $\contract$ guaranteeing a principal's utility of at least $\Omega(1)\cdot(1-\sum_{i}\alpha_i^\star)f(S^\star)$  for any CCE of $\contract$.
\begin{algorithmic}[1]
\State Let $\actions_0 = \{j \in \actions \mid c_j = 0\}$ be the zero cost actions
\If{$f(A_0)\geq \frac{2}{17}(1-\sum_i \alpha_i^\star) f(S^\star) $}
\State \Return $\contract=\vec{\epsilon} $ for $\epsilon = \frac{1}{2n}$
\ElsIf{  $\max_i \alpha^\star_{i} > 3/4 $}
\State Let $i^\star=\arg\max_i \alpha^\star_{i}$ \Comment{There must be a unique maximum}
\If{$(1-\alpha_i^\star) \cdot f(S_i^\star) \geq 4 \cdot f(S_{-i}^\star)$}
\State \Return $\contract$ where $\alpha_{i^\star} = \frac{1+\alpha_{i^\star}^\star}{2} $ and $\alpha_i=0$ for $i\neq i^\star$
\Else
\State \Return $\contract$ where $\alpha_{i^\star} = 0  $ and $\alpha_i=2\alpha_i^\star$ for $i\neq i^\star$
\EndIf
\Else
\State Partition $N$ into two bundles $B_1,B_2$, where for each $\ell \in\{1,2\}$, $\sum_{i\in B_\ell} \alpha_i^\star \leq 3/4$
\State Let $\ell \in \arg\max_{\ell'\in\{1,2\}} f(\bigcup_{i\in B_{\ell'}}S_i^\star)$
\State \Return $\contract$ where $\alpha_{i} = \frac{7}{6}\alpha_i^\star  $ for $i\in B_\ell$ and $\alpha_i=0$ for $i\not\in B_\ell$
\EndIf
\end{algorithmic}
\end{algorithm}

\begin{proof}[Proof of Theorem~\ref{thm:constant-gap-sub-CCE}]
Consider any input $\contract^\star$ and PNE $S^\star$ with respect to $\contract^\star$, given as input to Algorithm~\ref{alg:robust-submod}. We analyze the guarantee provided by the contract computed by this algorithm by distinguishing the following cases:

\bigskip

\noindent \textbf{Case A:} It holds that $f(A_0) \geq \frac{2}{17}(1-\sum_i \alpha_i^\star) f(S^\star)$. In this case, since $A_0$ is a PNE with respect to the zero-contract $\vec{0}$ (a contract that pays nothing to all agents), %then by 
by applying Lemma~\ref{lem:doubling-cce} with $\epsilon=\frac{1}{2n}$ and $\gamma =2$ we obtain that under contract $\contract=\gamma \vec{0} + \vec{\epsilon} = \vec{\epsilon} $ any CCE $\dist$ with respect to $\contract$ satisfies $\E_{S \sim \dist}[f(S)] \geq \frac{1}{2} (1 - \frac{1}{\gamma}) f(A_0) = \frac{1}{4} f(A_0)  $. 
Thus, the principal's utility under this case is at least $$(1-\sum_i \alpha_i)\E_{S \sim \dist}[f(S)] \geq (1-\frac{1}{2})\frac{1}{4} f(A_0) \geq \frac{1}{68}(1-\sum_i \alpha_i^\star)f(S^\star),$$ which concludes this case.

\bigskip

\noindent \textbf{Case B:} There exists an agent $i$ with  $\alpha^\star_{i} > 3/4 $ and $(1-\alpha_i^\star) \cdot f(S_i^\star) \geq 4 \cdot f(S_{-i}^\star)$ (and we are not in Case A).
In this case, since $\alpha^\star_i > 3/4$, we have $f(S_i^\star) \geq 16 \cdot f(S_{-i}^\star)$. 
    By subadditivity of $f$, this implies that $f(S_i^\star) \geq 16 \cdot (f(S^\star) - f(S_i^\star))$, or equivalently, 
    \begin{equation}
    f(S_i^\star) \geq \frac{16}{17} \cdot f(S^\star).\label{eq:1617}
    \end{equation}  
    
    Consider contract $\contract$ with $\alpha_i= \frac{1 + \alpha^\star_{i}}{2}$ and $\alpha_{j} = 0$ for $j\neq i$. Note that $\alpha_{i} > \alpha^\star_i$. 
    Let $\dist $ be any CCE of $\contract$. 
    We next show that $\E_{S\sim \dist}[f(S)] \geq \frac{1}{4} \cdot f(S_i^\star)$.
   
    We have the following: 
    \begin{eqnarray}        
     \E_{S\sim \dist}[\alpha_if(S_i) - c(S_i)] &\geq& \E_{S\sim \dist}[\alpha_i (f(S_i \cup S_{-i}) - f(S_{-i})) - c(S_i)]  \nonumber \\ & \geq & 
    \E_{S\sim \dist}[\alpha_i (f(S_i^\star\cup S_{-i}) - f(S_{-i})) ]- c(S_i^\star) \nonumber \\ & \geq &   \alpha_i f(S_i^\star) -\alpha_i\E_{S\sim \dist}[ f(S_{-i}) ]- c(S_i^\star)   \nonumber \\ & \geq &  \alpha_i f(S_i^\star) - \alpha_if(A_0)- c(S_i^\star),
    \label{eq:one-new}
    \end{eqnarray}
    where the first inequality is by subadditivity, the second inequality is since $\dist$ is a CCE of $\contract$, the third inequality is by monotonicity of $f$, and the last inequality is since no agent but $i$ will take non-zero cost actions in a CCE since $\alpha_j=0$ for $j\neq i$.
    On the other hand, since $S^\star$ is a PNE of $\contract^\star$, it must hold that 
    \[
    \alpha_i^\star f(S_i^\star \cup S_{-i}^\star) - c(S_i^\star) \geq \E_{S\sim \dist}[\alpha_i^\star f(S_i \cup S_{-i}^\star) - c(S_i)].
    \]
    By subadditivity of $f$, we have $f(S^\star_i \cup S^\star_{-i}) - f(S^\star_{-i}) \leq f(S^\star_i)$. Moreover, by monotonicity of $f$, we have  $f(S_i \cup S^\star_{-i}) \geq f(S_i)$. We thus obtain
    \begin{align}
    \alpha_i^\star f(S_i^\star) - c(S_i^\star) \geq \E_{S\sim \dist }[\alpha_i^\star f(S_i)  - c(S_i)] - \alpha_i^\star f(S_{-i}^\star).
    \label{eq:two-new}
    \end{align}
    By summing up~\eqref{eq:one-new} and~\eqref{eq:two-new}, we obtain
    \begin{align*}
    &\E_{S\sim \dist}[\alpha_if(S_i) - c(S_i)] + \alpha_i^\star f(S_i^\star) - c(S_i^\star) \\
    &\qquad\geq \E_{S\sim \dist }[\alpha_i^\star f(S_i)  - c(S_i)] - \alpha_i^\star f(S_{-i}^\star) +\alpha_i f(S_i^\star) - \alpha_if(A_0)- c(S_i^\star).
    \end{align*}
    Using linearity of expectation and that $\alpha_i \leq 1$ and $\alpha^\star_i \leq 1$, we get 
    \begin{equation}
    (\alpha_i - \alpha_i^\star) \E_{S\sim \dist}[f(S_i)] \geq (\alpha_i - \alpha_i^\star) \cdot  f(S_i^\star) - \alpha_i f(A_0) -\alpha_i^\star f(S_{-i}^\star) \geq (\alpha_i - \alpha_i^\star) \cdot  f(S_i^\star) -  f(A_0) - f(S_{-i}^\star) . \label{eq:alphaalpha-new}
    \end{equation}
    Thus, using that $\alpha_i > \alpha_i^\star$, 
    \begin{eqnarray}
    \E_{S\sim\dist}[f(S_i)]  & \geq & f(S_i^\star) - \frac{f(S_{-i}^\star) +f(A_0)}{\alpha_i - \alpha_i^\star}  = f(S_i^\star) - \frac{2(f(S_{-i}^\star) +f(A_0))}{1 - \alpha_i^\star}  \nonumber 
    \\ & \geq & f(S_i^\star) - \frac{2}{1 - \alpha_i^\star}\cdot  \left(\frac{1-\alpha_i^*}{4} \cdot f(S_{i}^\star)+ \frac{2}{17}(1-\sum_j \alpha_j^\star) f(S^\star) \right) \nonumber \\ &\geq& 
    f(S_i^\star) - \frac{2}{1 - \alpha_i^\star}\cdot  \left(\frac{1-\alpha_i^*}{4} \cdot f(S_{i}^\star)+ \frac{2}{17}(1- \alpha_i^\star) \frac{17}{16}f(S_i^\star) \right) = \frac{1}{4} f(S_{i}^\star)  \label{eq:s'1} .
    \end{eqnarray}
where the first inequality is by rearranging Inequality~\eqref{eq:alphaalpha-new}, the first equality is by the definition of $\alpha_i$, the second inequality is by the assumption of the case, and the last inequality is by Inequality~\eqref{eq:1617}.

 Overall, the principal's utility under contract $\contract$ and CCE $\dist$ is
    \begin{align*}
    \bigg(1 - \sum_{j} \alpha_{j}\bigg) \cdot \E_{S\sim \dist}[f(S)] & \geq (1-\alpha_i) \cdot \E_{S\sim \dist}[f(S_i)] = \frac{1}{2}(1-\alpha_i^\star) \cdot \E_{S\sim \dist}[f(S_i)]\geq 
    \frac{1}{8}(1-\alpha_i^\star) \cdot f(S_i^\star) \\
    &\geq \frac{1}{8} \bigg(1 - \sum_{j} \alpha^\star_{j}\bigg) \cdot f(S^\star_i) \geq \frac{2}{17} \bigg(1 - \sum_{j} \alpha^\star_{j}\bigg) \cdot f(S^\star),
    \end{align*}
    where the first equality is by definition of $\contract$ and monotonicity, the equality follows by the definition of $\alpha_i$, the second inequality holds by Inequality~\eqref{eq:s'1}, and the final inequality is by  Inequality~\eqref{eq:1617}
This concludes the argument for this case.

    \bigskip

    \noindent \textbf{Case C:} There exists an agent $i$ with $\alpha_{i} > 3/4 $ and $(1-\alpha_i^\star) \cdot f(S_i^\star) \leq 4 \cdot f(S_{-i}^\star)$ (and we are not in Case A).
    In this case, let $\contract$ be the contract where $\alpha_i=0$ and $\alpha_j=2\alpha_j^\star$ for $j\neq i$.
    By applying Lemma~\ref{lem:doubling-cce} on $S^\star$, $\contract^\star$, $N' = [n]\setminus \{i\}$ and $\gamma=2$ %and 
    we get that any CCE $\dist$ with respect to contract $\contract$ satisfies \begin{equation}
        \E_{S\sim\dist}[f(S)] \geq \frac{1}{4} f(\bigcup_{j \in N'} S_j^\star) . \label{eq:s'-new}
        \end{equation}
     We can bound the principal's utility under $\contract^\star$, and $S^\star$ by 
     \begin{eqnarray}
     \left(1-\sum_j \alpha_j^\star\right)f( S^\star) & \leq & \left(1-\sum_j \alpha_j^\star\right)f( S^\star_i) + \left(1-\sum_j \alpha_j^\star\right)f( S^\star_{-i})   
     \nonumber \\ & \leq & \left(1- \alpha_i^\star\right)f( S^\star_i) + f( S^\star_{-i})   \leq 5 \cdot f( S^\star_{-i}), \label{eq:s-i-new}
     \end{eqnarray}
     where the first inequality is by subadditivity, the last inequality is by the assumption of the case.

     One the other hand, under contract $\contract$ (for which $\sum_j \alpha_j = \sum_{j\neq i} 2\alpha^\star_j \leq  2(1- \alpha^\star_i) \leq \frac{1}{2}$), and CCE $\dist$, the principal's utility is  $$(1-\sum_j \alpha_j)\E_{S\sim \dist} [f(S)] \stackrel{\eqref{eq:s'-new}}{\geq }  (1- \sum_{j\neq i } 2\alpha^\star_j) \cdot   \frac{1}{4} f(\bigcup_{j \in N'} S_j^\star) \geq \frac{1}{8}   f( S_{-i}^\star) \stackrel{\eqref{eq:s-i-new}}{\geq} \frac{1}{40}   \left(1-\sum_j \alpha_j^\star\right) f( S^\star) ,$$
     which concludes the proof of the case.

    \bigskip

\noindent \textbf{Case D:} $\alpha_i^\star \leq 3/4$ for every $i$ (and we are not in case A).
Note that we can partition the agents into two bundles $B_1,B_2$, where for each $\ell \in\{1,2\}$, $\sum_{i\in B_\ell} \alpha_i^\star \leq 3/4$ (by the same argument as in the proof of Theorem~\ref{thm:constant-gap}). 
Now, assume without loss of generality that \begin{equation}
f(\bigcup_{i\in B_1} S_i^\star) \geq f(\bigcup_{i\in B_2} S_i^\star). \label{eq:assum-new}
\end{equation}
Let $\contract$ be the contract where $\alpha_i=0$ for $i\in B_2$ and $\alpha_i=\frac{7}{6}\alpha_i^\star$ for $i\in B_1$.
    By applying Lemma~\ref{lem:doubling-cce} on $S^\star$, $\contract^\star$, $N' = B_1$ and $\gamma=7/6$ we get that there any CCE  $\dist$ with respect to contract $\contract$ satisfies \begin{equation}
        \E_{S\sim\dist }[f(S)] \geq \frac{1}{14} f(\bigcup_{i \in B_1} S_j^\star) \label{eq:s'2-new}
        \end{equation}

The principal's utility from $\contract$ and $\dist$ is 
\begin{eqnarray*}
     (1-\sum_i \alpha_i) \E_{S\sim \dist}[f(S)]  & \geq & (1-\frac{7}{6} \sum_{i\in B_1} \alpha_i^\star) \cdot \frac{1}{14}  \cdot f(\bigcup_{i \in B_1} S_i^\star)  \geq (1-\frac{7}{6} \cdot \frac{3}{4})  \cdot\frac{1}{14} \cdot  f(\bigcup_{i \in B_1} S_i^\star) \\ &  \geq  & \frac{1}{112} \frac{f(\bigcup_{i \in B_1} S_i^\star)+f(\bigcup_{i \in B_2} S_i^\star)}{2}  \geq \frac{f(S^\star)}{224},
     \end{eqnarray*}
where the first inequality is by the definition of $\contract$ and by Inequality~\eqref{eq:s'2-new}, the second inequality is since $\sum_{i\in B_1} \alpha^\star_i\leq 3/4$, the third inequality is by Inequality~\eqref{eq:assum-new}, and the last inequality is by subadditivity.
This concludes the proof of the theorem.
\end{proof}

%% PROOF OF PROPOSITION 4.5

\section{Proof of Proposition~\ref{prop:subadditive-upper-bound}}~\label{app:proof-subadditive}%
We first prove a weaker variant of the Scaling-for-Existence Lemma for subadditive functions.

\begin{lemma}[Scaling-for-Existence Lemma for Subadditive]\label{lem:double-sub}
Suppose $f$ is subadditive. Let $\dist$ be a dropout-stable distribution with respect to $\contract$. For any set of $i \in N$ and $\gamma > 1$, let $\contract'$ be defined by $\alpha'_i = \gamma \cdot \alpha_i$,  and $\alpha'_{i'} = 0$ for $i'\neq i$.
Then, there exists a pure Nash equilibrium $S'$ with respect to $\contract'$, satisfying  $$f(S') \geq (1- \frac{1}{\gamma}) \cdot \E_{S\sim \dist}[f( S_i)].$$
\end{lemma}
\begin{proof}
    Fix a set $S$ in the support of $\dist$. Now for any set $S_i'\subseteq A_i$ define $\Phi(S_i', \contract) = f(S_i') - \frac{c(S_i')}{\alpha_i}$ where $\frac{0}{0}$ is interpreted as $0$, and $\frac{c}{0}$ for a positive $c$ is interpreted as $\infty$. 
    Then we have
    \begin{equation}
    \E_{S\sim \dist}[\Phi(S_i, \contract)]  = \E_{S \sim \dist}[f(S_i) -  \frac{c(S_i)}{\alpha_i}] \geq  \E_{S\sim \dist }[f(S_i \mid S_{-i}) -  \frac{c(S_i)}{\alpha_i}] \geq 0 , \label{eq:upp-potential}
    \end{equation}
    where the first inequality is by subadditivity, and the last inequality is by dropout-stability.

    Let $S'_i\subseteq A_i$ be a set of actions maximizing  $\Phi(S'_i, \contract')$ (when fixing  $\contract'$ as defined in the statement of the lemma). Then we have  
    \begin{align*}
    \Phi(S'_i, \contract') \geq \E_{S\sim \dist}[\Phi(S_{i},  \contract')] &= \E_{S\sim \dist }\left[f(S_{i}) -  \frac{c(S_i)}{\gamma \alpha_i}\right]\\
    &= \left(1-\frac{1}{\gamma}\right) \E_{S \sim \dist}[f(S_{i})] + \frac{1}{\gamma} \E_{S\sim \dist}[\Phi(S_{i},\contract)] \geq \left(1-\frac{1}{\gamma}\right) \E_{S\sim \dist}[f(S_{i})],
    \end{align*}
    where the first inequality follows by the maximality of $S'_i$, and the last inequality follows by Eq.~\eqref{eq:upp-potential}.

    As $S'_i$ is a global maximum of $\Phi(\cdot, \contract')$, it is also a local maximum. 
    Since $\Phi(\cdot,\contract')$ is a potential function for the game induced by the contract $\contract'$ (see Proposition~\ref{prop:potential}), this means that $S'_i$ is a pure Nash equilibrium with respect to contract $\contract'$. 
\end{proof}

We are now ready to prove Proposition~\ref{prop:subadditive-upper-bound}.
Consider any contract $\contract^\star$ and any coarse-correlated equilibrium $\dist$ with respect to $\contract^\star$. In the remainder of the proof, all expectations are over $S^\star$ that is distributed according to $ \dist$. 
We consider three cases:
\bigskip

\noindent \textbf{Case A:} There exists an agent $i$ with  $\alpha^\star_{i} > 3/4 $ and $(1-\alpha_i^\star) \cdot \E[f(S_i^\star)] \geq 4 \cdot \E[f(S_{-i}^\star)]$.
The proof of this case is identical to the proof of Case A of the proof of Theorem~\ref{thm:constant-gap}. 

    \bigskip

    \noindent \textbf{Case B:} There exists an agent $i$ with $\alpha^\star_{i} > 3/4 $ and $(1-\alpha_i^\star) \cdot \E[f(S_i^\star)] \leq 4 \cdot \E[f(S_{-i}^\star)]$.
    In this case, let $i^\star=\arg\max_{i'\neq i} \E[f(S_{i'})]$. Let $\contract'$ be the contract where $\alpha_{i'}'=0$ for $i'\neq i^\star$, and $\alpha_{i^\star}'=2\alpha_{i^\star}$.
    By applying Lemma~\ref{lem:double-sub} on $\dist$, $\contract^\star$, $i^\star$ and $\gamma=2$ we get that there exists an equilibrium $S'$ with respect to contract $\contract'$ such that \begin{equation}
        f(S') \geq \frac{1}{2} \Ex{f( S_{i^\star}^\star)} . \label{eq:s'-new3}
        \end{equation}
     We can bound the expected principal's utility under $\contract^\star$, and $\dist$ by 
     \begin{eqnarray}
     \left(1- \sum_j \alpha_j^\star\right)\Ex{f( S^\star)} & \leq & \left(1-\sum_j \alpha_j^\star\right)\Ex{f( S^\star_i)} + \left(1-\sum_j \alpha_j^\star\right)\Ex{f( S^\star_{-i})}   
     \nonumber \\ & \leq & \left(1- \alpha_i^\star\right)\Ex{f( S^\star_i)} + \Ex{f( S^\star_{-i})}   \leq 5 \cdot \Ex{f( S^\star_{-i})}, \label{eq:s-i-new3}
     \end{eqnarray}
     where the first inequality is by subadditivity, the last inequality is by the assumption of the case.

     One the other hand, under contract $\contract'$ (for which $\sum_j \alpha'_j =  2\alpha_{i^\star}^\star \leq  2(1- \alpha_{i}^\star) \leq \frac{1}{2}$), and equilibrium $S'$, the principal's utility is  
     $$(1-\sum_j \alpha_j') f(S') \stackrel{\eqref{eq:s'-new3}}{\geq }  (1-  2\alpha_{i^\star}^\star) \cdot   \frac{1}{2} \Ex{f( S_{i^\star}^\star)} \geq \frac{1}{4n}   \Ex{f( S_{-i}^\star)} \stackrel{\eqref{eq:s-i-new3}}{\geq} \frac{1}{20n}   \left(1- \sum_j \alpha_j^\star\right)\Ex{f( S^\star)} ,$$
     where the second inequality is by subadditivity and the definition of $i^\star$.
     This concludes the proof of the case.

    \bigskip

\noindent \textbf{Case C:} $\alpha_i^\star \leq 3/4$ for every $i$.
    In this case, let $i^\star=\arg\max_{i} \E[f(S_{i})]$. Let $\contract'$ be the contract where $\alpha_{i}'=0$ for $i\neq i^\star$, and $\alpha_{i^\star}'=\frac{7}{6}\alpha_{i^\star}$.
    By applying Lemma~\ref{lem:double-sub} on $\dist$, $\contract^\star$, $i^\star$ and $\gamma=\frac{7}{6}$ we get that there exists an equilibrium $S'$ with respect to contract $\contract'$ such that \begin{equation}
        f(S') \geq \frac{1}{7} \Ex{f( S_{i^\star}^\star)} . \label{eq:s'-new2-subadd}
        \end{equation}

     Under contract $\contract'$ (for which $\sum_j \alpha'_j =  \frac{7}{6}\alpha_{i^\star}^\star \leq  \frac{7}{6} \cdot \frac{3}{4} = \frac{7}{8}$), and equilibrium $S'$, the principal's utility is  
     $$(1-\sum_j \alpha_j') f(S') \stackrel{\eqref{eq:s'-new2-subadd}}{\geq }  (1-  \frac{7}{8}) \cdot   \frac{1}{7} \Ex{f( S_{i^\star}^\star)} = \frac{1}{56}   \Ex{f( S_{i}^\star)} \geq \frac{1}{56n}   \Ex{f( S^\star)} ,$$
     where the last inequality is by the definition of $i^\star$ and by subadditivity.
This concludes the proof of the proposition. \qed

\end{document}